\documentclass[12pt,a4paper,fleqn]{scrartcl}
\RequirePackage{amsthm,amsmath,natbib}
\RequirePackage{amssymb,amstext}
\RequirePackage{hypernat}
\usepackage{graphicx}
\usepackage{enumerate}
\usepackage{bbm}
\usepackage{url}
\usepackage{verbatim}
\usepackage{ bm}
\usepackage[pdfpagemode=UseOutlines ,plainpages=false
,hypertexnames=false ,pdfpagelabels ,hyperindex=true,colorlinks=true]{hyperref}
	\makeatletter%
	\Hy@breaklinkstrue%
	\makeatother%
\usepackage{color}
\definecolor{darkred}{rgb}{0.6,0.0,0.1}
\definecolor{darkgreen}{rgb}{0,0.5,0}
\definecolor{darkblue}{rgb}{0,0,0.5}
\hypersetup{colorlinks ,linkcolor=darkblue ,filecolor=darkgreen
,urlcolor=darkblue ,citecolor=black}
\usepackage{ bm}

\renewcommand{\cite}{\citet}
\renewcommand{\Delta}{D}
\usepackage{graphicx}
\newcommand{\Ex}{{\mathbb E}}
\newcommand{\Var}{ {\mathbb V}ar}
\newcommand{\set}[1]{{\left\lbrace #1\right\rbrace }}	
\def\1{\mathop{\mathbbm 1}\nolimits}

\newcommand{\PP}{\mathbb{P}}

\newcommand{\sol}{\varphi}
\newcommand{\cF}{{\cal F}}
\newcommand{\cX}{{\cal X}}
\newcommand{\cN}{{\cal N}}

\newcommand{\cH}{{\cal H}}

\newtheorem{prop}{Proposition}[section]
\newtheorem{coro}[prop]{Corollary}
\newtheorem{theo}[prop]{Theorem}

\newtheorem{lem}[prop]{Lemma}
\newtheorem{rem}{Remark}[section]
\newtheorem{example}{Example}[section]

\newtheorem{assA}{Assumption}
 
\numberwithin{equation}{section}

 \author{\textsc{  Christoph Breunig}\thanks{Department of Economics, Emory University, Rich Memorial Building, Atlanta, GA 30322, USA, e-mail:
\url{christoph.breunig@emory.edu}}\\
{\small \textit{Emory University}}
  \and\textsc{ Peter Haan}\thanks{DIW Berlin  and Freie Universit\"at Berlin, Mohrenstr. 58, 10117 Berlin, Germany, e-mail: \url{ phaan@diw.de}}\\
  {\small \textit{DIW Berlin and Freie Universit\"at Berlin}}
}

 \title{Nonparametric Regression with Selectively Missing Covariates\thanks{The authors gratefully thank the Co-Editor Jianqing Fan and two anonymous referees for their many constructive comments on the previous version of the paper. Support by German Science Foundation through CRC TRR 190 is gratefully acknowledged.}
 }

\begin{document}
   \maketitle
\begin{abstract}
\vskip -.3cm
{\footnotesize We consider the problem of regression with selectively observed covariates
 in a nonparametric framework. Our approach relies on instrumental variables that explain variation in the latent covariates but have no direct effect on selection. The regression function of interest is shown to be a weighted version of observed conditional expectation where the weighting function is a fraction of selection probabilities. Nonparametric identification of the  fractional probability weight (FPW) function is achieved via a partial completeness assumption. We provide primitive functional form assumptions for partial completeness to hold. The identification result is constructive for the FPW series estimator. We derive the rate of convergence and also the pointwise asymptotic distribution. In both cases, the asymptotic performance of the FPW series estimator does not suffer from the inverse problem which derives from the nonparametric instrumental variable approach. In a Monte Carlo study, we analyze the finite sample properties of our estimator and we compare our approach to inverse probability weighting, which can be used alternatively for unconditional moment estimation.
In the empirical application, we focus on two different applications. We estimate the association between income and health using linked data from the SHARE survey and administrative pension information and use pension entitlements as an instrument. In the second application we revisit the question how income affects the demand for housing based on data from the German Socio-Economic Panel Study (SOEP). In this application we use regional income information on the residential block level as an instrument.  In both applications we show that income is selectively missing and we demonstrate that standard methods that do not account for the nonrandom selection process lead to significantly biased estimates for individuals with low income.}
\end{abstract}
{\small \begin{tabbing}
 \> nonparametric identification, partial completeness, incomplete data, \\
 \> series estimation, income distribution, health, housing.\\[.2ex]
\end{tabbing}}
\section{Introduction}

Sample selection is a central challenge for empirical evaluation studies. Nonrandom selection can affect the empirical analysis in many ways, for example through nonrandom selection into treatment programs, selective measurement error or through selective nonresponse or missingness of data.
In this paper, we propose an instrumental variable approach to address the problem of nonrandom selection.
 We provide constructive nonparametric identification results and build on those to establish a novel fractional probability weighting estimator.

While the methodology is general and applicable to many situations in which selection might be problematic the leading example in this paper will be selective nonresponse and selective missing data.
We are interested in the identification and estimation of the nonparametric regression function $g(x)=\Ex[Y|X^*=x]$ where $Y$ is always observed but $X^*$ is only selectively observed. In this case, parts of the information of the covariates are missing not at random for some sampling units. Without accounting for selectivity of responses, statements about individual behavior based on such incomplete data might be severely biased.

In this paper, we establish identification of the nonparametric regression function $g(x)=\Ex[Y|X^*=x]$  based on instrumental variables that explain variation in the latent covariates but have no direct effect on selection. Such an instrumental variable approach is well suited when selection is driven by the latent variables $X^*$. We show that the regression function $g$ can be written as a weighted version of its observed counterpart. The weighting function is determined by a fraction of selection probabilities, i.e., fractional probability weights (FPW), that depends on latent variables. We propose a novel identification restriction, the so called \textit{partial completeness assumption}, which implies identification of the FPW and thus of the nonparametric regression function $g$. In contrast to usual completeness assumptions required for identification of nonparametric instrumental variable models, we are able to provide primitive, functional form conditions for the partial completeness assumption to hold. We emphasize that these functional form conditions do not imply (semi-)parametric restrictions but only impose a nonparametric structure in different forms of separability of $Y$ and $X^*$. Specifically, we show that a nonparametric generalized additive structure of the selection probability is sufficient to obtain identification of the function $g$.

Based on the constructive nonparametric identification result we propose a novel nonparametric FPW series estimator that is convenient for implementation.
We show that our estimator has a rate of convergence that coincides with usual nonparametric regression estimators, i.e., the asymptotic performance of the estimator is the same as of an estimator with full information of the underlying selection mechanism.
We establish asymptotic normality of the estimator  and show that the asymptotic variance is not necessarily enlarged by FPW estimation.
We also propose a bootstrap procedure to construct uniform confidence bands. A Monte Carlo simulation study demonstrates the improvements of our approach over missing at random (MAR) estimators. In particular, we highlight our contribution also in a finite sample analysis of linear regression with alternative inverse probability weighting estimators (IPW).

Finally, we use the method in two different empirical applications. Both applications are important for the discussion about income inequality and highly relevant for public policy. First, we use the developed methodology to analyze the association between income and the risk of bad health. The empirical analysis is based on linked data from the Survey of Health, Aging and Retirement in Europe (SHARE). We exploit a specific feature of the data which allows us to link a sub-sample of the survey data to administrative data of the German pension insurance.\footnote{\cite{BinMar_17} compare self reported information about income and education from SHARE  with matched information from Danish administrative data in order to study the implications of measurement error.} We find that income in the SHARE data is selectively missing and we demonstrate that standard methods that do not account for the nonrandom selection process are biased, specifically for individuals with low incomes. Assuming linearity of the regression function $g$  the point estimate of income is significantly negative when imposing MAR however it is not significantly different from zero when accounting for nonrandom nonresponse.
In the second example we analyze how housing varies with income. We quantify the relationship between labor earnings and the probability to own a house. For this empirical analysis we use data from the German Socio Economic Panel (SOEP) and exploit  information about the regional purchasing power collected on the residential block level (Sub-Zip code level) as an instrument. Again we find that earnings are selectively missing and we demonstrate that the estimates derived in standard methods are  biased for individuals with earnings in the lowest decile which is a specifically relevant group for public policy. For individuals with higher earnings the estimates from the different methods do not differ significantly.

Our paper is linked to several strands of the literature.  The most common way to deal with missing data is to assume missing at random pioneered by \cite{little2002}. In the context of selectively missing covariates, a sieve semiparametric  maximum likelihood estimator was proposed by \cite{chen2007}.
In contrast, an instrumental variable strategy, as proposed in the paper, was used so far only to deal with endogenous missingness of dependent variables, see, for instance,  \cite{tang2003}, \cite{ramalho2013}, and \cite{2010Hault}. Also \cite{breunig2015} consider the problem of nonparametric regression with selective nonresponse of the dependent variables; \cite{zhao2015} focus on a semiparametric approach. There only has been minor attention to selectively observed covariates. One example is \cite{zhaostatistica} who consider a semiparametric approach to deal with selectively missing covariates that is crucially different from ours. While \cite{zhaostatistica} require a parametric specification of the distribution of outcome given potential covariates, we leave these conditional distribution unrestricted. We establish nonparametric identification of the regression function and hence ensure that the identification is not due to specific functional form restrictions that might be violated in practice.

The paper adds as well to the literature on income inequality, more specifically to studies on the income gradient on health outcomes and mortality e.g. \cite{Pre_1975}, \cite{DeatPax_98}, \cite{CutDeaLLe_06}, or \cite{CutLleVog_11}, and on the effect  of income on housing and on housing demand, see e.g.  \cite{QuiRap2004}, \cite{Albouyetal2016} or \cite{DusFitZim2018}. In general these studies are based on survey data in which wealth, income, health and housing  information and further demographic variables are self reported. As shown in \cite{breunig2017} information on income or earnings in surveys is likely to suffer from nonrandom selection which might result in biased estimates of the association between income and health or income and home ownership. In this respect this study extends the previous literature as we account for nonrandom nonresponse of the income information.

The paper is organized as follows. In Section \ref{model_identification} we establish identification of our nonparametric model. In Section \ref{s_regression_estimation} we derive the FPW series estimator, establish its rate of convergence and its asymptotic normality, and derive uniform bootstrap confidence bands. Section \ref{s_Monte_Carlo_simulation} provides Monte Carlo simulations and discusses implications of FPW estimation to unconditional moments.
In Section \ref{s_empiricalapplication} we apply our methodology to the two different empirical applications. All proofs can be found in Appendix \ref{app:proofs}. Appendix \ref{app:tech} provides some technical results. Finally, Appendix \ref{sec:mis:dep} provides an extension when also the dependent variable is selectively missing.

\section{Nonparametric Identification}\label{model_identification}
This section consists of two subsections. In Subsection \ref {ss_setup}, we provide assumptions required for identification. In particular, we introduce a novel restriction, i.e., the partial completeness assumption, and provide primitive conditions for it. Subsection \ref{ss_ident} establishes identification of the nonparametric regression function.

\subsection{Setup and Main Assumptions}\label{ss_setup}
 Given an observable outcome  variable $Y$ and latent covariates $X^*$ our interest lies in the regression function $g(x)=\Ex[Y|X^*=x]$.
Identification relies on instrumental variables $W$ that explain variations of the latent variable $X^*$ but are not directly related to the selection mechanism $\Delta$. This is formalized in the following.
Throughout the paper, we assume that a sample $(\Delta_1,Y_1,W_1),\dots,(\Delta_n,Y_n,W_n)$ of $(\Delta,Y,W)$ is observed for each individual. A $d_x$-- dimensional vector of covariates $X^*$ is only fully observed depending on a binary indicator variable $\Delta$, i.e., $X^*$ is observed when $\Delta=1$ and missing when $\Delta=0$.
We write $X=\Delta X^*$.\footnote{The situation can be easily extended to a multivariate version where $\Delta$ denotes a $d_x$--dimensional vector of missing data indicators. In order to keep the notation simple we do not treat this case explicitly.} Under the  assumptions  presented below we see that the selection probability conditional on $(Y,X^*)$, i.e., $\PP(\Delta=1|Y,X^*)$, is only partially identified but  still point identification of the regression function $g$ is established.
    \begin{assA}[Exclusion Restriction]\label{A_instruments}
      It holds that
        \begin{align*}
      \PP(\Delta = 1|Y, X^*,W) = \PP(\Delta = 1|Y, X^*).
        \end{align*}
    \end{assA}
Assumption \ref{A_instruments} states an exclusion restriction of the random vector $W$.
It excludes any relation between $W$ and the selection mechanism $\Delta$ that is not channeled through $(Y,X^*)$.  The setting corresponds to the measurement error set up, where instrumental variables are required to drive the latent, true variable but not the variable that is observed with error.
 However, identification with nonclassical measurement error requires an additional  exclusion restriction which restricts $W$ to have no information on $Y$ that is not captured in $X^*$, see Assumption 2 $(i)$ in \cite{hu2008}. Interestingly, nonrandom selection as extreme form of nonclassical measurement error simplifies the exclusion restriction imposed on the instruments.

We also emphasize that Assumption \ref{A_instruments} allows for dependence of $D$ and $Y$. Thus,  our approach captures selection on unobservables that do not only stem from latent characteristics in $X^*$ but also from unobservables  that are unexplained by the regression function $g(X^*)$.\footnote{In the model $Y=g(X^*)+U$,  not only $X^*$ but also unobservables $U$ are allowed to directly affect the selection mechanism $D$.} This is an important feature of our framework, as in many economic environments, selection variables can be driven by unobserved individual characteristics. Related literature on nonrandom nonresponse of covariates does not allow for such a general selection mechanism, see \cite{zhao2015}.
We introduce the function class $\mathcal B =\{\phi:\, \Ex|\phi(Y,X^*)|<\infty \text{  and }\inf_{y,x}\phi(y,x)>0\}$.
\begin{assA}[Partial Completeness]\label{A_identification}
For all $\phi\in\mathcal B$ it holds: $\Ex[\phi(Y,X^*)|Y,W]=0$  implies that $\phi$ does not depend on $Y$.
\end{assA}
Assumption \ref{A_identification} is less restrictive than the usual completeness assumption which assumes that $\Ex[\phi(Y,X^*)|Y,W]=0$  implies $\phi(Y,X^*)=0$. This assumption is commonly imposed to ensure identification in nonparametric instrumental variable models, see for instance \cite{NP03econometrica}. In the context of endogenous selection such completeness assumptions were considered by \cite{2010Hault} and \cite{breunig2015}.
 On the other hand, the partial completeness assumption holds under mild functional form assumptions as shown below.

Assumption \ref{A_identification} is automatically satisfied if $\phi$ does not depend on $Y$.
Indeed, if the selection probability $\PP(\Delta=1|Y,X^*)$ does not depend on $Y$ the regression function $g$ is identified as we see in the next subsection and  thus, the partial completeness assumption is well suited for our particular selection problem. Moreover, the next result provides functional form restriction under which partial completeness holds. Throughout the paper, $f_V$ denotes the probability density function of a random variable $V$.
\begin{prop}\label{prop:prim}
Assume that $f_{X^*|YW}=f_{X^*|W}$. Assume that for any $\phi\in\mathcal B$ there exist functions $\phi_1$ and $\phi_2$ such that either
\begin{align}
\phi(Y,X^*)&=\phi_1(Y) \phi_2(X^*)-1\label{eq:mult}
\end{align}
or
\begin{align}
\phi(Y,X^*)&=\psi\big(\phi_1(Y)+\phi_2(X^*)\big)-1\label{eq:tf}
\end{align}
where $\psi$ is an analytic function with $\Ex[\psi^{(j)}\left(\phi_1(y)+\phi_2(X^*)\right)|W]\neq 0$ for some $j\geq 1$ and $\Ex[\sum_{j\geq 1}|\psi^{(j)}\left(\phi_1(y')+\phi_2(X^*)\right)\left(\phi_1(y)-\phi_1(y')\right)^j/(j!)|]<\infty$ for all $y, y'$ in the support of $Y$.
%
Then, Assumption \ref{A_identification}  is satisfied.
\end{prop}
Proposition \ref{prop:prim} requires that $Y$ does not provide information on $X^*$ that is not contained in the vector $W$. Given this mild restriction we see from Proposition \ref{prop:prim} that functional form restrictions imply the partial completeness assumption to hold.
We emphasize that these functional form restrictions do not imply (semi-)parametric specifications but only impose a nonparametric structure in different forms of separability of $Y$ and $X^*$.
Note that we subtract by one as the exclusion restriction in Assumption \ref{A_instruments} implies the conditional mean restriction $\Ex[\Delta/\PP(\Delta=1|Y,X^*)-1|Y,W]=0$. The selection probability $\PP(\Delta=1|Y,X^*)$ is not necessarily point identified through the former conditional mean restriction given Assumption \ref{A_identification}. Equation \eqref{eq:tf} provides a  generalized additive type restriction on the functions of interest. 
The restriction $\Ex[\psi^{(j)}\left(\phi_1(y)+\phi_2(X^*)\right)|W]\neq 0$, for some integer $j\geq 1$ and all $y$,  is satisfied by a broad class of link functions $\psi$. This condition holds, in particular, when the first derivative $\psi'\left(\phi_1(y)+\phi_2(X^*)\right)$ is strictly positive for all $y$ in the support of $Y$. Consequently, we conclude from Proposition \ref{prop:prim} that partial completeness holds under a generalized additive restriction when the link function $\psi$ coincides with any analytic, cumulative distribution function supported on the whole real line.
Another case where $\Ex[\psi^{(j)}\left(\phi_1(y)+\phi_2(X^*)\right)|W]\neq 0$, for some $j\geq 1$,  holds is when $\psi$ is a non-constant polynomial function.

While Proposition \ref{prop:prim} provides a broad class of functions which satisfy the partial completeness assumption, partial completeness can fail for functions which vary in $y$ with vanishing means and when instruments are independent of $X^*$. This is demonstrated in the following illustrative example.
\begin{example}
Consider the case where $\Var(X^*)=1$,  $\Ex[X^*]=0$, and, as in Proposition 2.1,  $f_{X^*|YW}=f_{X^*|W}$. In addition, we assume that the instrumental variables $W$ have no information on $X^*$, i.e., $f_{X^*|W}=f_{X^*}$.
Consider the function $\phi(y,x)=x^2-yx-1$, then we obtain:
\begin{align*}
\Ex[\phi(Y,X^*)|Y=y,W=w]&=\Var(X^*)-y\Ex[X^*]-1\\
&=\Var(X^*)-1\\
&=0,
\end{align*}
but $\phi$ varies in $y$ and hence, partial completeness does not hold.
\end{example}

\begin{assA}\label{A_pos}
 The selection probability  $\PP(\Delta=1|Y,X^*)$ is bounded away from zero uniformly over its support.
\end{assA}
Assumption \ref{A_pos} can rule out a selection when it is a deterministic function of $Y$ and $X^*$, such as certain indicator functions. We also emphasize that Assumption \ref{A_pos} can be relaxed if we are only interested in point evaluation at some $x_0$ in the support of $X$. That is, identification of $\Ex[Y|X^*=x_0]$ requires only $\PP(\Delta=1|Y=y,X^*=x_0)>0$ uniformly over all $y$ in the support of $Y$.

\subsection{Nonparametric Identification via FPW Weighting}\label{ss_ident}
In this subsection, we establish identification of the nonparametric regression function $g(x)=\Ex[Y|X^*=x]$.
We show that the function $g$ can be identified via a fractional probability weight (FPW). In addition, we show that the FPW is identified by making use of instrumental variables $W$ which satisfy the previous assumptions.
In the next result, we document that the regression function $g$ can be written as
\begin{align}\label{eq:ident:g}
g(x)=\Ex\left[Y\, \omega(Y, x)\Big|\Delta=1, X^*=x\right]
\end{align}
where the fractional probability weight (FPW) function $\omega$ is given by
\begin{align}\label{def:omega}
\omega(y, x)=\frac{\PP(\Delta=1|X^*=x)}{\PP(\Delta=1|Y=y, X^*=x)}.
\end{align}
Further, we establish identification of the nonparametric regression function $g$.
\begin{theo}\label{thm:ident}
Let Assumptions \ref{A_instruments}--\ref{A_pos} be satisfied. Then, the FPW  function $\omega$ is identified and thus,  identification of the regression function $g$ follows  through \eqref{eq:ident:g}.
\end{theo}
The previous result shows that the FPW function given in \eqref{def:omega} is point identified although the selection probabilities conditional on latent variables are only partially identified.
This is an implication of partial completeness imposed in Assumption \ref {A_identification}.
Corollary \ref{coro:fpw}  presents a useful property of the FPW function $\omega$. This result is an immediate consequence of the proof of Theorem \ref{thm:ident} and hence we omit its proof.
\begin{coro}\label{coro:fpw}
Let  Assumption \ref{A_pos} be satisfied. Then, for the FPW  function $\omega$ we obtain
\begin{align*}
\Ex[\omega(Y, x)|\Delta=1,X^*=x]=1.
\end{align*}
\end{coro}
In empirical applications also the dependent variable might be selectively missing.  We discuss this case in Appendix \ref{sec:mis:dep}. The remark below highlights the difference between selectively missing covariates which requires FPW and selectively missing dependent variables which requires inverse probability weighting.

\begin{rem}[Relation to selectively missing depend variables]\label{rem:fpw:ipw}
Consider the case of selectively missing dependent variables, that is, $Y^*$ is only observed if $D^Y=1$ and otherwise missing, while $X$ is always observed. Further, assume that $\PP(D^Y=1|Y^*,X)=\PP(D^Y=1|Y^*)$.
In this case, (as shown by  \cite{breunig2015}) we obtain
\begin{align*}
\Ex[Y^*|X]=\Ex[Y\psi(Y)|X]
\end{align*}
where $Y=Y^* D^Y$ and $\psi(\cdot)=1/\PP(D^Y=1|Y^*=\cdot)$. Consequently, the conditional expectation with latent $Y^*$ corresponds to an inverse probability weighted version of observed counterparts. This correction differs from  fractional probability weighting, as considered in this paper, which is required for selectively missing covariates.
\end{rem}
While Remark \ref{rem:fpw:ipw} highlights the difference between fractional and inverse probability weighting, we note that for unconditional moment estimation both approaches are feasible (see Section \ref{sec:ipw:fpw}). Still our FPW approach is preferable in this case due to a less restrictive identification requirement and finite sample improvements as shown in Section \ref{sec:ipw:fpw}.
\section{The FPW Series Estimator and its Asymptotic Properties}\label{s_regression_estimation}
This section consists of four subsections. In Subsection \ref{ss_semiparametric_estimation}, we derive the FPW series estimator which stems from our constructive identification result. Subsection \ref{subsec:rate} provides the rate of convergence of the estimator. We establish pointwise asymptotic normality of the FPW series estimator in Subsection \ref{subsec:inf} and provide asymptotic validity of uniform bootstrap confidence bands in Subsection \ref{subsec:ucb}.

\subsection{Estimation}\label{ss_semiparametric_estimation}
 We define the conditional selection probability by
\begin{equation}\label{par:restr}
  \varphi(y,x):=\PP(\Delta=1|Y=y,X^*=x).
\end{equation}
In particular, the selection probability conditional on the latent regressors $X^*$ is determined by
\begin{align*}
 \PP(\Delta=1|X^*=x )
&=\left(\Ex \left[\frac{1}{ \varphi(Y,x)}\Big|\Delta=1, X^*=x\right]\right)^{-1},
\end{align*}
see the proof of Theorem \ref{thm:ident}.
We thus obtain the following expression for the FPW function $\omega$:
\begin{align*}
\omega(y, x)=\left(\Ex \left[\frac{\varphi(y,x)}{ \varphi(Y,x)}\Big|\Delta=1, X^*=x\right]\right)^{-1}.
\end{align*}

 We estimate the regression function $g(x)=\Ex[Y|X^*=x]=\Ex[Y\omega(Y,x)|\Delta=1, X^*=x]$ using a plug-in series least squares estimator.
To do so,  we introduce a vector of basis functions $p^K(\cdot)=\big(p_1(\cdot)\dots,p_K(\cdot)\big)'$;  $K=K(n)$  is an integer which increases with the sample size $n$. We further introduce the $n\times K$--matrix  $\mathbf  X=\big(\Delta_1 p^K(X_1),\dots,\Delta_n p^K(X_n)\big)$.
We estimate the FPW function $\omega$ via
\begin{align*}
\widehat \omega(y, x;\phi)=\left(p^K(x)'\,(\mathbf  X'\mathbf X)^{-1}\,\sum_{i=1, \Delta_i=1}^n p^K(X_i) \frac{\phi(y, x)}{\phi(Y_i, X_i)}\right)^{-1}.
\end{align*}
It is  common in the context of inverse probability weighting, to normalize the weights to sum up to one. In our context, we normalize $\omega$ as follows.
Employing Corollary \ref{coro:fpw}, i.e., $\Ex[\omega(Y, x)|\Delta=1,X^*=x]=1$, we obtain
\begin{align*}
\Ex[\Delta\, p^K(X)p^K(X)']&=\Ex[\Delta\, p^K(X^*)p^K(X^*)']\\
&= \Ex[\Delta\, p^K(X)\, \omega(Y,X)\, p^K(X)'].
\end{align*}
Replacing $\Ex[\Delta\, p^K(X)\, \omega(Y,X)\, p^K(X)']$ by the empirical matrix $\mathbf X'\, \boldsymbol \omega(\widehat \varphi)\,\mathbf X$ we obtain the FPW series estimator of the regression function $g$ given by
\begin{equation}\label{gen:def:est}
\widehat g(x)\equiv \,
p^K(x)'\,\big(\mathbf X'\, \boldsymbol \omega(\widehat\varphi)\,\mathbf X\big)^{-1}\,\sum_{i=1, \Delta_i=1}^n p^K(X_i) Y_i\, \widehat\omega(Y_i, X_i;\widehat\varphi)
\end{equation}
where $\boldsymbol \omega(\phi)=\text{diag}\big(\widehat \omega(Y_1,X_1;\phi),\dots,\widehat \omega(Y_n,X_n;\phi)\big)$. Here, $\widehat\varphi$ is a restricted sieve minimum distance estimator of the selection probability $\varphi$ given as follows. We have the conditional moment restriction induced by the exclusion restriction imposed in Assumption \ref{A_instruments}, that is,
\begin{align}\label{cond:eq}
\Ex\left[\frac{\Delta}{\varphi(Y,X)}\Big|Y,W\right]=1.
\end{align}
(Here, we use that $\varphi(Y,X^*)=\varphi(Y,X)$ whenever $\Delta=1$.)
Consider  a vector of tensor product basis functions $q^L(\cdot,\cdot)=\big(q_1(\cdot,\cdot)\dots,q_L(\cdot,\cdot)\big)'$ used to approximate the conditional mean in \eqref{cond:eq};  $L=L(n)$  is an integer which increases with the sample size $n$. We estimate the conditional mean $m(y,w; \phi)=\Ex[\Delta/\phi(Y,X)-1|Y=y,W=w]$ by the series least squares estimator
\begin{align*}
\widehat m(y,w; \phi):=q^L(y,w)\Big(\sum_{i=1}^nq^L(Y_i,W_i)q^L(Y_i,W_i)'\Big)^{-1}\sum_{i=1}^n q^L(Y_i,W_i)\Big(\frac{D_i}{\phi(Y_i,X_i)}-1\Big)
\end{align*}
then we the constrained  sieve minimum distance estimator
\begin{align}\label{est:prob}
\widehat\varphi=\underset{\phi\in \mathcal B_L}{\mathrm{argmin}}\sum_{i=1}^n \widehat m^2(Y_i,W_i; \phi)
\end{align}
where $\mathcal B_L=\{\phi=1/(\beta'q^L):\, \min_{y,x}\beta'q^L(y,x)\geq 1\}$ following \cite{chenpouzo2012}. We may assume that $\mathcal B_L$ becomes dense in $\mathcal B$ as $L$ tends to infinity.
Note that in the case of nonparametric estimation, any estimator of $\varphi$ has a slow rate of convergence since the conditional mean restriction yields in general to a so called ill-posed inverse problem, see \cite{NP03econometrica} and \cite{BCK07econometrica}. In our case, $\varphi$ is not identified through the conditional mean equation \eqref{cond:eq} but we can always ensure uniqueness of the estimator, for instance, by considering the minimal norm estimator of equation \eqref{cond:eq}.
Finally, note that FPW series estimation is convenient since control variables which enter the model linearly can be simply included in the empirical matrix $\mathbf X$. This allows to treat partially linear models as considered in our empirical applications in Section~\ref{s_empiricalapplication}.

\subsection{Rate of Convergence}\label{subsec:rate}
We now introduce some assumptions. The support of $X$ is denoted by $\mathcal X$. We also introduce the $L^2_X$--norm $\|\phi\|_X=\sqrt{\Ex\phi^2(X)}$ and $\|\cdot\|$ denotes the Euclidean norm. We make use of the notation $U=Y-g(X^*)$ and   $h(x,\phi)=\Ex \left[1/\phi(V)|\Delta=1, X^*=x\right]$. Recall that $\omega(V,\varphi)=(\varphi(V)h(X,\varphi))^{-1}$ is identified due to Theorem \ref{thm:ident}.
 We introduce function class
\begin{align*}
\cH=\set{\psi:\,\psi(\cdot,\phi)\in L_X^2\text{ and }\psi(\cdot,\phi)\geq 1\text{ for all }\phi\in\mathcal B}.
\end{align*}
Note that $h(\cdot,\varphi)\in\cH$ due to Assumption \ref{A_pos}.  In order to achieve the rate of convergence, we have to restrict the complexity of the function classes $\mathcal B$ and $\mathcal H$, by imposing a finite entropy integral. We first provide primitive conditions by imposing smoothness restrictions on the function classes under consideration.

For any vector $a=(a_1,\dots, a_d)$ of $d$ integers, where $d=\dim(X,Y)$, define $D^a = \partial^{|a|}/\partial a_{a_1}\dots \partial x_{a_d}$, where $|a|=\sum_{l=1}^d a_l$. Let $\mathcal R_d$ be a bounded, convex subset of $\mathbb R^d$ with nonempty interior. For a function $\phi: \mathcal R_d\to\mathbb R$ and some $\alpha>0$, let $\underline\alpha$ be the largest integer smaller than $\alpha$, and
\begin{align*}
\|\phi\|_{\infty,\alpha}=\max_{|a|\leq \underline\alpha} \sup_v |D^a\phi(v)|+\max_{|a|= \underline\alpha} \sup_{v\neq v'}
\frac{|D^a\phi(v)-D^a\phi(v')|}{\|v-v'\|^{\alpha-\underline\alpha}}.
\end{align*}
We further restrict the function class $\mathcal B$ in the next assumption to contain only functions with bounded $\|\cdot\|_{\infty,\alpha}$ norm and bounded supremum norm $\|\cdot\|_\infty$. 
To do so, we introduce
\begin{align}\label{def:B:smooth}
\mathcal B_\alpha= \{\phi: \mathcal R_d\to\mathbb R\text{ continuous with }\inf_{v\in\mathcal R_d} \phi(v)>0 \text{  and }\|\phi\|_{\infty,\alpha}<C\}
\end{align}
for some $\alpha>0$ and constant $C>0$.
Similarly for $\mathcal H$ we introduce a smooth function class
\begin{equation*}
\mathcal H_\alpha=\{\psi:\mathcal R_{d_x}\to\mathbb R\text{ continuous with} \,\|\psi(\cdot,\phi)\|_{\infty,\alpha}<C\text{ and }\psi(\cdot,\phi)\geq 1\text{ for all }\phi\in\mathcal B\}
\end{equation*}
for some $\alpha>0$ and constant $C>0$, which is a subset of $\mathcal H$ for any constant $C>0$.

\begin{assA}\label{Ass_bas}
(i) We observe a sample $((\Delta_1,Y_1, X_1, W_1),\dots,(\Delta_n,Y_n,X_n,W_n))$ of independent and identical distributed (i.i.d.) copies of $(\Delta,Y,X,W)$ where $X=\Delta X^*$.
 (ii) There exist a constant $C>0$ and a sequence of positive integers $K:=K(n)$ satisfying  $\sup_{x\in\cX}\|p^K(x)\|^2\leq C K$ such that $K^2/n=o(1)$.
 (iii) The smallest eigenvalue of $\Ex[\Delta\, p^K(X)p^K(X)']$ is bounded away from zero uniformly in $K$. (iv) Let $\Ex[U^2|X^*]<\infty$, $\Ex[g^2(X^*)]<\infty$, and $\|\gamma'p^K-g\|_X=O(K^{-\alpha/d_x})$ for some $\gamma\in\mathbb R^K$.
 (v) 
 $\mathcal B\subset\mathcal B_\alpha$ with $\alpha>d/2$ and $\mathcal H\subset\mathcal H_{\alpha'}$ with $\alpha'>d_x/2$.
\end{assA}
Assumption \ref{Ass_bas} $(ii)-(iii)$ restricts the magnitude of the approximating functions $\{p_j\}_{j\geq1}$ and imposes nonsingularity of their second moment matrix.
Assumption \ref{Ass_bas} $(ii)$ holds for instance for polynomial splines, Fourier series and wavelet bases.
Assumption \ref{Ass_bas} $(iii)$ is satisfied if $p^K$ is a vector of orthonormal basis functions and the probability density function of $X^*$ given $D=1$ is uniformly bounded away from zero on its support.
Assumption \ref{Ass_bas} $(iv)$ determines the sieve approximation error which in turn characterizes the bias of the estimated regression function $g$; see also \cite{chen2007} for further discussions on sieve bases.
Assumption \ref{Ass_bas} $(v)$ imposes smoothness conditions on the function classes $\mathcal B$ and $\cH$ in order to reduce the complexity of these classes.

The next result establishes the rate of convergence for the FPW series estimator $\widehat g$.
\begin{theo}\label{thm:est:cond:par}
Let Assumptions \ref{A_instruments}--\ref{Ass_bas}  be satisfied. Then, we have
\begin{equation*}
\|\widehat g-g\|_X^2=O_p\Big(\max\big(K^{-2\alpha/{d_x}}, \frac{K}{n} \big)\Big).
\end{equation*}
\end{theo}
 From Theorem \ref{thm:est:cond:par} we observe that the estimator $\widehat g$ attains the usual bias and variance term in integrated mean square error for nonparametric series regression. If $K$ is chosen to level variance and bias, i.e., $K\sim n^{d_x/(2\alpha+d_x)}$,  then the convergence rate given in Theorem \ref{thm:est:cond:par} coincides with $n^{-2\alpha/(2\alpha+d_x)}$. Consequently, we obtain the optimal nonparametric rate of convergence as in the situation where the covariates $X$ are completely observed, that is, we do not obtain a slower rate of convergence due to the estimation of the FPW function. This result is obtained by the regularity conditions imposed in Assumption \ref{Ass_bas} $ (v)$, which restrict the complexity of the underlying fractional probability function $\omega$.

\subsection{Pointwise Inference}\label{subsec:inf}
This subsection discusses the inference of the estimator of regression function $g$ evaluated at some point of the support of $X$. In applications, such asymptotic distribution results can be useful to construct approximate confidence intervals. Before stating the result we make the following additional assumptions, in particular, with respect to the error term $U=Y-g(X^*)$.

\begin{assA}\label{A:inf:par} (i) $\Ex[|U\Delta\, \omega(Y,X^*)|^4]$ is bounded from above and $\Var(U\Delta\,\omega(Y,X^*)|X^*)$ is uniformly bounded away from zero. (ii) For some $x$ in the support of $X$ it holds $\sum_{j=1}^K|p_j(x)|=O(\|p^K(x)\|)$.
\end{assA}

The bounds imposed in Assumption \ref{A:inf:par} $(i)$ are not stronger than the one imposed in \cite{Newey1997}. We also note that it is possible to relax these conditions as noted by \cite{belloni2015} or \cite{chen2015optimal}. Assumption \ref{A:inf:par} $(ii)$ is a condition on the basis functions and satisfied for B-splines or wavelets, see Appendix E of \cite{chen2018optimal}.

To obtain asymptotic normality of our estimator we require a normalization factor. Therefore, we introduce the sieve variance given by
\begin{equation*}
 \textsl{v}_K(x)=p^K(x)'\Ex\Big[p^K(X^*)\Var(U \Delta\,\omega(Y,X^*)|X^*)\,p^K(X^*)'\Big]\, p^K(x).
\end{equation*}
In contrast to the usual series regression in \cite{Newey1997}, we see that the sieve variance also contains the FPW function $\omega$. As $\omega$ can take values smaller than one, the sieve variance for our FPW series estimator can be even smaller than the one associated to the usual series estimator. This is in contrast to estimators based on weighting via inverse selection probabilities that always lead to larger sieve variances, see \cite{breunig2015} for selective outcomes or \cite{das2003} for propensity score weighting.
We replace the sieve variance $ \textsl{v}_K(x)$ by the estimator
\begin{equation*}
\widehat{\textsl{v}}_K(x)=p^K(x)'(\mathbf X'\mathbf X)^{-1}n^{-1}\sum_{i=1, \Delta_i=1}^np^K(X_i)\widehat U_i^2\,\Delta_i \,\widehat\omega^2(Y_i,X_i;\widehat\varphi)\,p^K(X_i)'\,(\mathbf X'\mathbf X)^{-1} p^K(x)
\end{equation*}
where $\widehat U_i=Y_i-\widehat g(X_i)$.
We now establish the asymptotic distribution of the estimator $\widehat g$ evaluated at some point $x$ in the support of $X$. Similarly, asymptotic distribution results for linear functionals of $g$ can be obtained. We introduce the supremum norm $\|\phi\|_\infty=\sup_{x\in\mathcal X}|\phi(x)|$.
\begin{theo}\label{thm:inference:par}
  Let Assumptions \ref{A_instruments}--\ref{A:inf:par} be satisfied. If for some $x$ in the support of $X$ it holds
\begin{equation}\label{cond:inference:par}
 n\,\|\gamma'p^K-g\|_\infty^2=o(\textsl{v}_K(x))
\end{equation}
then we have
    \begin{align*}
      \sqrt{n/\, \widehat{\textsl{v}}_K(x)}\,\big(\widehat g(x) - g(x)\big)\stackrel{d}{\rightarrow}\mathcal N(0,1).
    \end{align*}
\end{theo}
Condition \eqref{cond:inference:par} requires the estimator of $g$ to be undersmoothed. This ensures that the sieve approximation bias in the second step estimation procedure becomes asymptotically negligible.
Theorem \ref{thm:inference:par} can be also used to construct pointwise confidence intervals for $g(x)$ but can also be extended to construct uniform bootstrap confidence bands, as the following remark illustrates.

\subsection{Uniform Bootstrap Confidence Bands}\label{subsec:ucb}
This section provides a  bootstrap procedure to construct uniform confidence bands for $g$ and establishes asymptotic validity of it.
Let $(\varepsilon_1,\ldots,\varepsilon_n)$ be a bootstrap sequence of i.i.d. random variables drawn independently of the data $\{(\Delta_i,Y_i,X_i,W_i)\}_{1\leq i\leq n}$, with $\Ex[\varepsilon_{i}] = 0$, $\Ex[\varepsilon_{i}^2] = 1$, with bounded moments.
Common choices of distributions for $\varepsilon_i$ include the standard Normal, Rademacher, and the two-point distribution of \cite{mammen1993}. Further, $\mathbb{P}^*$ denotes the probability distribution of the bootstrap innovations $(\varepsilon_1,\ldots,\varepsilon_n)$ conditional on the data.
We introduce  the bootstrap process
\begin{align*}
 \mathbb Z^B(x) =  \frac{p^K(x)'(\bold X'\boldsymbol\omega(\widehat\varphi)\bold X/n)^{-1}}{\sqrt{n\widehat{\textsl{v}}_K(x)}}\sum_{i=1, \Delta_i=1}^np^K(X_i)\Big(Y_i\,\widehat \omega(Y_i,X_i;\widehat \sol) - \widehat g(X_i)\Big)\varepsilon_{i}.
\end{align*}
Under regularity conditions it can be shown that the bootstrap process provides a uniform approximation of the influence function of the estimator $\widehat g$ and thus, can be used to construct uniform confidence bands (see also \cite{chen2018optimal}).

For the implementation of a $100(1-\alpha)\%$ uniform confidence bands, we compute the critical value $z^B_{1-\alpha}$ as the $(1-\alpha)$ quantile of $\sup_{x}| \mathbb Z^B(x)|$ for a large number of independent bootstrap draws. The resulting $100(1-\alpha)\%$ uniform confidence band is then given by
\begin{align*}
x\mapsto \left[\widehat g(x)-z^B_{1-\alpha}\sqrt{\frac{\widehat{\textsl{v}}_K(x)}{n}}, \widehat g(x)-z^B_{1-\alpha}\sqrt{\frac{\widehat{\textsl{v}}_K(x)}{n}}\right].
\end{align*}

In the following, we assume that $\mathcal C$ is  a closed subset of $\mathbb R^{d_x}$.
Let $\delta(\cdot,\cdot)$ be the standard deviation semimetric on $\mathcal{C}$ of the Gaussian Process $\mathbb Z(x) = p^K(x)'\mathcal Z/\sqrt{\textsl{v}_K(x)}$ with $\mathcal Z\sim \mathcal N(0,\Sigma)$ and $\Sigma=\Ex[p^K(X^*)\Var(U \Delta\,\omega(Y,X^*)|X^*)\,p^K(X^*)']$ defined as $\delta(x_1,x_2) = (\Ex[(\mathbb Z(x_1) - \mathbb Z(x_2))^2])^{1/2}$, see \textit{e.g.} \cite[Appendix A.2]{Vaart2000}. Further, we introduce the notation $N(\varepsilon,\mathcal{F},\|\cdot\|_{\mathcal F})$ for the $\epsilon$-entropy of $\mathcal{F}$ with respect to a norm $\|\cdot\|_{\mathcal F}$.
\begin{assA}\label{Ass:uniform}
  (i) $\mathcal{C}$ is compact and $(\mathcal{C},\delta(\cdot,\cdot))$ is separable for each $n\geq 1$. (ii) There exists a sequence of finite positive integers $c_n$ such that
$ 1 + \int_0^{\infty}\sqrt{\log N(\epsilon, \mathcal{C},\delta(\cdot,\cdot))}d\epsilon = O (c_n)$.
(iii) There exists a sequence of positive integers $r_n$ with $r_n = o(1)$ such that $K^{5/2}= o(r_n^3\sqrt{n})$, $r_n c_n=O(1)$,  and
\begin{equation*}
  K^2\sqrt{\frac{\log(n)}{n}}+ K n^{-1/4}\log(n)\Big(c_n+\sup_{x\in\mathcal C}\sqrt{\frac{n}{\textsl{v}_K(x)}}|\gamma'p^K(x)-g(x)|\Big)=o(r_n).
   \end{equation*}
\end{assA}
Assumption \ref{Ass:uniform} is similar to
\cite[Assumption 6]{chen2018optimal} who establish asymptotic validity of uniform confidence bands in nonparametric instrumental variable estimation.
Assumption \ref{Ass:uniform} $(ii)$ is a mild regularity assumption, see also \cite[Remark 4.2]{chen2018optimal} for sufficient conditions.

The next theorem establishes asymptotic validity of the bootstrap for constructing uniform confidence bands for the regression function $g$.
Below, ${\PP}^*$ denotes a probability measure conditional on the data $\{(\Delta_i, Y_i,X_i,W_i)\}_{i=1}^n$.
\begin{theo}\label{thm:bands}
  Let the assumptions of Theorem \ref{thm:inference:par} and Assumption \ref{Ass:uniform} hold.
  Then, we have
  \begin{equation*}
      \sup_{s\in\mathbb{R}}\left|\PP\left(\sup_{x\in\mathcal{C}}\left|\sqrt{\frac{n}{\widehat{\textsl{v}}_K(x)}}\Big(\widehat g(x)-g(x)\Big)\right|\leq s\right) - \PP^*\left(\sup_{x\in\mathcal{C}}\left|\mathbb Z^B(x)\right| \leq s \right)\right| = o_p(1).
  \end{equation*}
\end{theo}

\section{Monte Carlo simulation}\label{s_Monte_Carlo_simulation}
In this section, we study the finite-sample performance of our estimator by presenting the results of a Monte Carlo simulation. We first focus on the estimation of nonlinear conditional moments, then we turn to linear regression. We perform $1000$ Monte Carlo replications in each experiment and the sample size is $n=1000$.

\subsection{Nonlinear Regression}
We consider the estimation of the regression function $g$ under the following simulation design.
The data are generated by $W=\Phi(\xi)$ and $X^*=\Phi(\chi)$ where $\chi=\rho\,\xi+\sqrt{1-\rho^2}\,\nu$
and  $(\xi,\nu)'\sim\cN(0,I_2)$.
Here,  $\rho$ characterizes the strength of the instruments and is varied in the experiments below.
Further, we draw $Y$ from the model
 \begin{equation*}
Y=g(X^*)+U
\end{equation*}
where $g(x)=\Phi\big(c_g(x-0.5)\big)$ with standard normal distribution function $\Phi$, a constant $c_g\in\{5,20\}$,  and $U\sim \cN(0,\sigma_U^2)$, where $\sigma_U^2\in\{0.5,1\}$.
\begin{figure}[ht!]
		\includegraphics[width=16.5cm]{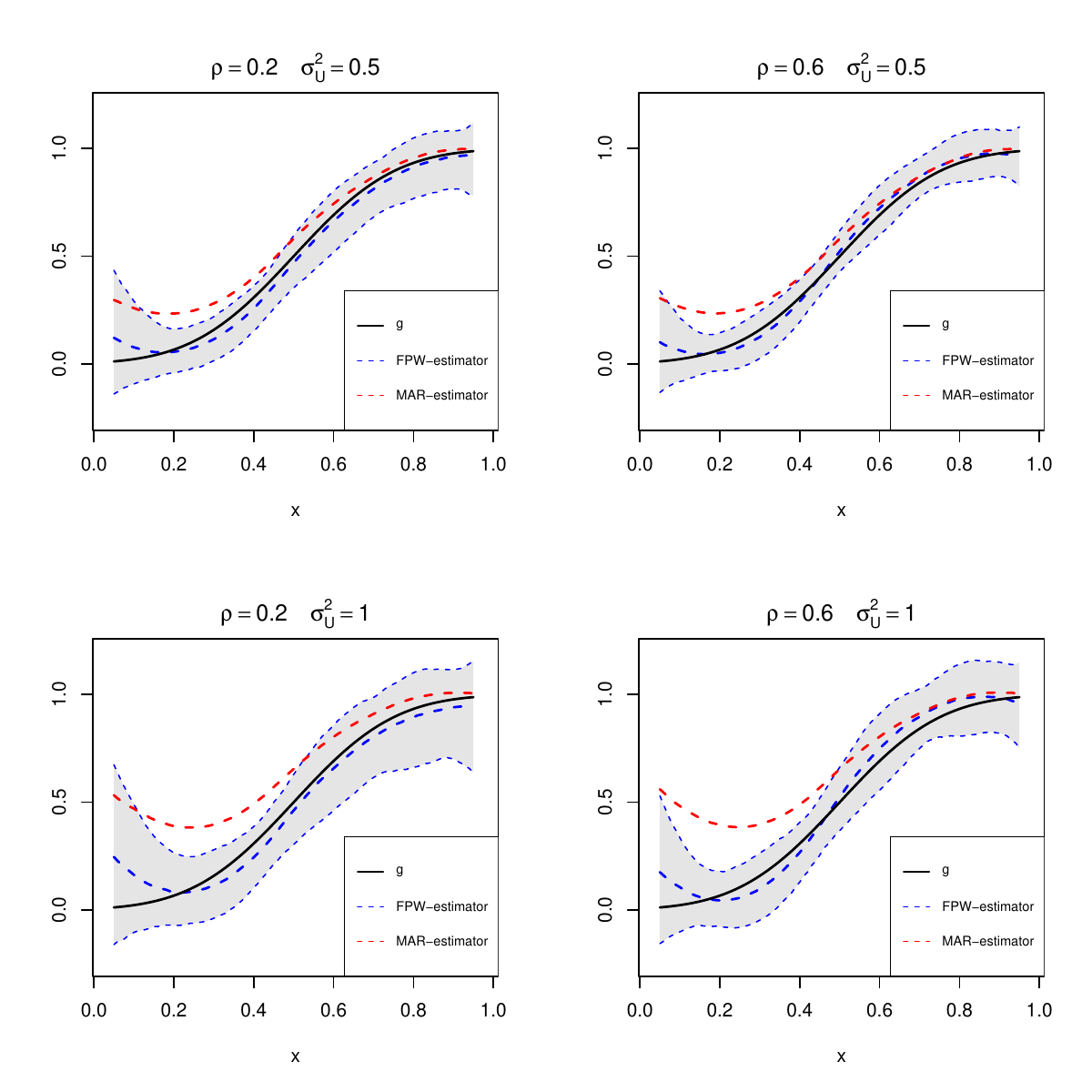}
	\caption{\small The regression function $g$ with $c_g=5$, the median of $\widehat g$ (blue) with 95\% confidence intervals and an estimator under MAR assumption (red).}\label{condmean}
\end{figure}
We generate realizations of the selection variable $\Delta$ from the Bernoulli distribution
\begin{align}\label{sel:prob:sim}
\Delta\sim\textsc{Bernoulli}\big(\Phi(1+\chi+U/2)\big).
\end{align}
Consequently, the selection probability is a function of the latent covariates $X^*$ and also unobservables $U$.
The selection probability $\sol$ is estimated  using the sieve minimum distance procedure in \eqref{est:prob} with tensor product of quadratic B-splines and zero knots (hence $L=9$).
We estimate the function $g$ by using the FPW series estimator $\widehat g$ given in \eqref{gen:def:est}.  As basis functions we use quadratic B-splines  with 2 knots (hence $K=5$) when $c_g=5$  (see Figure \ref{condmean}) and quadratic B-splines  with 7 knots (hence $K=10$) when $c_g=20$ (see Figure \ref{condmean-nonsmooth}). For the B-splines of $X$ note that we choose the interior knots located at the quantiles of the observed subset of $\{X_i\}_{D_i=1,1\leq i\leq n}$.

Figures \ref{condmean} and \ref{condmean-nonsmooth} depict the median of the FPW series estimator $\widehat g$ together with its 95\% pointwise confidence bands and a series estimator under the missing at random (MAR) assumption based on listwise deletion under different simulation designs. We vary the parameters $\rho$ and $\sigma_U$; the first row shows results with $\sigma_U^2=0.5$, the second row with $\sigma_U^2=1$, the first column with $\rho=0.2$, and the second column with $\rho=0.6$.
In all cases the median of the FPW series estimator is close to the true regression function $g$ and the MAR series estimator is severely biased.
From Figures   \ref{condmean} and \ref{condmean-nonsmooth} we see that the strength of instruments only has a moderate influence on the performance of the FPW series estimator. This is in line with our theoretical results that the asymptotic performance of the estimator is not driven by the correlation of the instruments to the latent covariates. On the other hand, we see that the variance of estimation becomes much larger as $\sigma_U^2$ increases from $0.5$ to $1$. For larger values of $\sigma_U$ the problem of selection on unobservables becomes more severe. In this case, the confidence intervals of the FPW series estimator become larger but also the bias of the MAR series estimator increases.
\begin{figure}[ht!]
		\includegraphics[width=16.5cm]{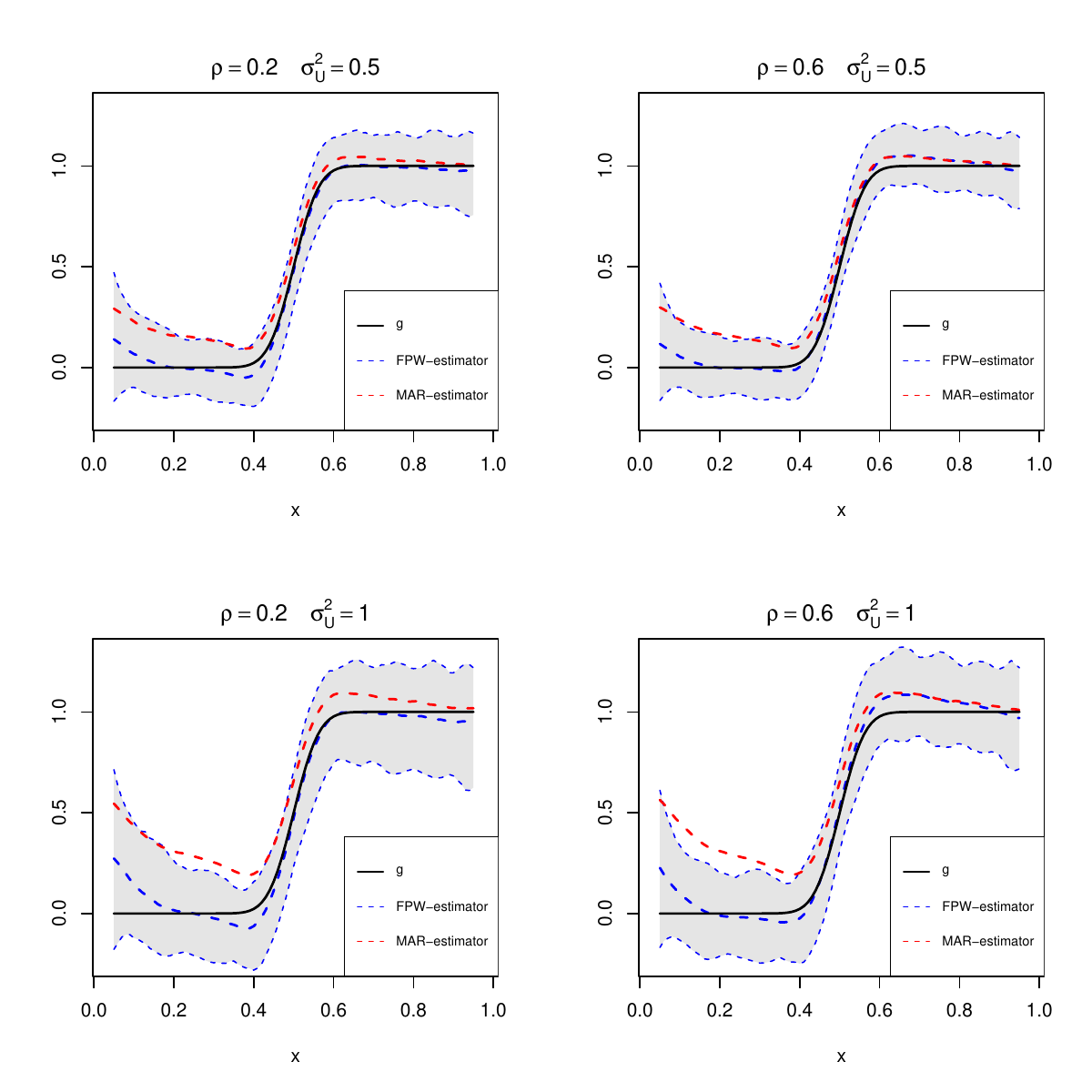}
	\caption{\small The regression function $g$ with $c_g=20$, the median of $\widehat g$ (blue) with 95\% confidence intervals and an estimator under MAR assumption (red).}\label{condmean-nonsmooth}
\end{figure}

\subsection{Linear Regression and Comparison to Inverse Probability Weighting (IPW)}\label{sec:ipw:fpw}
For estimators based on unconditional moments, an alternative approach to FPW is given by IPW. Yet this subsection demonstrates that the  FPW approach leads to more accurate estimation results even in linear regression models.

We generate the data as described in the previous subsection. As we are interested in the unconditional mean $\Ex[YX^*]$ we could also make use of IPW. Indeed, making use of the notation  $\widetilde \omega(y,x)=1/\PP(\Delta=1|Y=y,X^*=x)$ we obtain by the law of iterated expectations that
\begin{align*}
\Ex[Y X^*]&=\Ex[Y X^* \Delta/\PP(\Delta=1|Y,X^*)]\\
&=\Ex[Y X\widetilde \omega(Y,X)].
\end{align*}
Alternatively, we can apply the FPW function  $\omega(y, x)=\PP(\Delta=1|X^*=x)/\PP(\Delta=1|Y=y, X^*=x)$ and obtain by our nonparametric identification results that
\begin{align*}
\Ex[Y X^*]=\Ex[Y X\omega(Y,X)].
\end{align*}
Estimating the unconditional mean by FPW has the advantage over IPW that identification of the inverse selection probability is more restrictive than identification of the FPW function $\omega$. That is, for identification of the IPW the usual completeness assumption is required. In addition, we demonstrate in the following the finite sample properties of both approaches in a finite sample analysis. To do so, we consider the linear model
\begin{align}\label{sim:lin:mod}
Y=\beta_0+\beta_1 X^*+U
\end{align}
where $\beta_0=1$ and $\beta_1=3$. The data is generated as described in the previous subsection with $\rho=0.2$ and $\sigma_U=1$. Below, we analyze the absolute median bias and the coverage at the 95\% nominal coverage rate for the FPW estimator, the IPW estimator, the MAR estimator based on listwise deletion and the estimator when there is no missing data, that is, $D\equiv 1$. We estimate the weights for FPW and IPW nonparametrically as described in the previous subsection. The FPW and IPW estimators coincide then with weighted ordinary least squares (OLS) estimators.

\begin{table}[ht]	
\renewcommand{\arraystretch}{1.3}
  \begin{center}
\begin{tabular}{lcccc}
\hline\hline
 & FPW & IPW & MAR & $D\equiv 1$\\
\hline
Abs. median bias$(\widehat\beta_0)$ & 0.039 &  0.056&0.403&0.006\\
Abs. median bias$(\widehat\beta_1)$ & 0.112 &  0.147&0.447&0.008\\
 Coverage for $\beta_0$ & 0.873 &  0.799&0.004&0.943\\
Coverage for $\beta_1$ & 0.798 &  0.728&0.085&0.934\\
  \hline
  \end{tabular}
    \caption{{\small Absolute median bias and empirical coverage at 95\% nominal coverage rate for the linear model \eqref{sim:lin:mod} for the FPW estimator, the IPW estimator,  the MAR estimator, and the OLS estimator without missingness.}}\label{table:uncond}
  \end{center}
  \end{table}

In Table \ref{table:uncond} we compare the IPW and FPW estimators with the OLS estimator under the MAR hypothesis and the OLS estimator when there is no missing data. The second and third row show the absolute median bias of the estimators of the intercept and the slope parameter. We see that the FPW estimator has smaller median bias than the IPW estimator for both parameters. Not surprisingly, the bias dramatically increases when we ignore selection and consider the MAR estimator.
The last two rows depict the coverage of the confidence interval for the intercept and the slope parameter.  We see that the FPW estimator has more accurate coverage than the IPW estimator.
Yet there is undercoverage of the FPW estimator which is due to the severity of the selectivity of the nonresponse mechanism. Note that the 95\% confidence interval of the  MAR estimator contains the true intercept only in 4 out of 1000 Monte Carlo Iterations. If we relax the severity of the selectivity then the coverage of the FPW estimator is more accurate. For instance, if in \eqref{sel:prob:sim} the variable $U/2$ is replaced by $U/3$ then the empirical coverage of the FPW estimator for $\beta_1$ increases from 0.798 to 0.887.

\section{Empirical Applications}\label{s_empiricalapplication}

In the final section of the paper we apply the developed methodology to study to relevant economic questions. First, we focus on the association between income and health. In the second application we analyze how income affects the demand for housing.

\subsection{Application I: the association between income and health}\label{s_health}

As mentioned in the introduction, a large body of literature has  documented a positive correlation between income and health, see e.g. \cite{DeatPax_98}.\footnote{In general, it is difficult to identify the causal effect of income on health, therefore most studies focus on the association on income and health. We follow these studies. Notable exceptions are studies that focus on the effect of income or wealth shocks on health, see e.g. \cite{Schwand_18}. } However, in general these studies are based on survey data in which income, health information and further demographic variables are self reported. As shown e.g. in \cite{breunig2017} information on income in surveys is likely to suffer from nonrandom selection which might result in biased estimates of the association between income and health.

The empirical analysis is based on linked data from the German sample of the Survey of Health, Aging and Retirement in Europe (SHARE, Wave 5, collected in 2013) and the German pension insurance. SHARE is a multinational survey of the elderly population aged 50 and above in Europe, for more information see \cite{abs_2013}. The survey includes standard demographic characteristics and self reported information about different income measures and various subjective and objective health outcomes. The key variables for our analysis are individual income and health outcomes. We use a broad definition of income. For non-retired individuals the income includes labor earnings, income from self employment and transfers for unemployed. For retired individuals the income is composed of own pensions, and if applicable widowers pension and additional labor earnings. The health status is described by an objective measurement of the hand grip strength. Previous studies have documented that hand grip strength is a good measure of physical functioning and a predictor of morbidity, disability and mortality, see e.g. \cite{rantanen1999midlife}, \cite{bohannon2015muscle}, or \cite{dodds2014grip}. From the grip strength we construct a binary variable which indicates bad health status if the grip strength is below the 25th percentile.

For our analysis we exploit a specific feature of the data which allows us to link a subsample\footnote{The linkage of the data requires the consent of the individuals, about 2/3 of individuals agreed to the linkage.} of the survey data to administrative data of the German pension insurance. Thus, in addition to the self-reported income information which might suffer from nonrandom nonresponse the data includes official information about pension entitlements.  For pensioners we observe the full pension entitlements, i.e. number of pension points, they have earned during their working life; for non retired individuals we observe the entitlements they have collected so far. Pension entitlements are a deterministic function of the full individual earnings history. The earnings history is a good predictor of current income, however, it contains no direct information about the response behavior for current income. Therefore, this information allows us to construct a suitable instrument to account for potential nonrandom nonresponse of the current self reported income. In fact this instrument is superior to instruments based on self reported lagged employment outcomes which are often used, see e.g. \cite{breunig2017}. First, the instrument is not affected by transitory shocks since it combines information about the full working life instead of using information of only one period. Second, self reported past information might as well suffer from nonresponse. This is not the case for the information about pension entitlements in administrative data.

In the empirical analysis we concentrate on $3340$ individuals which are younger than 80 years and who have agreed to the linkage of the survey data and the information of the pension insurance. Out of this sample, $12.34\%$ do not respond to the income information question.\footnote{In our sample only 4.8\% do not provide information about grip strength - we assume that this information is missing at random. Using the test of missing (completely) at random by \cite{breunig2017} we obtain the value of the test statistic 0.053 with $0.05$-- level critical value of 0.10 and hence, we fail to reject the MAR hypothesis.} Table \ref{table:sum:health} provides summary statistics of the relevant variables for the analysis.
\bigskip

\begin{table}[h]
\renewcommand{\arraystretch}{1.3}
\centering
{\small
\begin{tabular}{lrrrrc}
\hline\hline
 & 1st Qu. & Median & Mean & 3rd Qu.   \\ \hline
\text{Grip strength}    & 28.00  & 35.00 & 36.53 &45.00  \\
\text{Bad health}    & 0  & 0 & 0.24 &0  \\
$\log(\text{Income per Year})$  & 8.96 & 9.48 & 9.37 & 9.95  \\
Number of pension points   & 19.90 & 33.20 & 33.59 & 47.00 \\
Age in years    &56.00 &  62.00 &  62.39  & 69.00

\\ \hline
\end{tabular}%
}

\caption{\small This table provides summary statistics of the relevant variables for the analysis. The sample includes $3340$ individuals which are aged between 50 and 80 years. Grip strength is measured in kilograms. The value of a pension point in 2013 amount to 24.92 (East Germany) and 28.07 (West Germany). Bad health is defined when grip strength is below the 25th percentile. }\label{table:sum:health}
\end{table}

To quantify the association between income and health we use the following semiparametric model 
\begin{align}\label{est_eq}
Badhealth_i= g\big(\log(Income_i^*)\big)+\alpha_0 Age_i +\beta_0 Gender_i + U_i,
\end{align}
where the function $g$ and the parameters $\alpha_0$ and $\beta_0$ are unknown. We assume that $U_i$ is conditional mean independent of the explanatory variables, $\log(Income_i^*)$, $Age_i$, and $Gender_i$. We apply the FPW estimator as described in the previous section, i.e., we estimate the nonparametric selection probability  using quadratic B-spline basis functions for least square approximations. Specifically, the selection probability $\sol$ is estimated  using the sieve minimum distance procedure described in \eqref{est:prob} with tensor product of quadratic B-splines and zero knots, where we additionally control for age and gender.
We estimate the function $g$ using the FPW series estimator $\widehat g$ given in \eqref{gen:def:est} using quadratic B-splines  with 1 knot (placed at the median of observed income) and again controlling for age and gender. The MAR estimator uses the same choice of B-spline basis functions with the same knot placement.

\begin{figure}[ht]
	\centering
		\includegraphics[width=8cm]{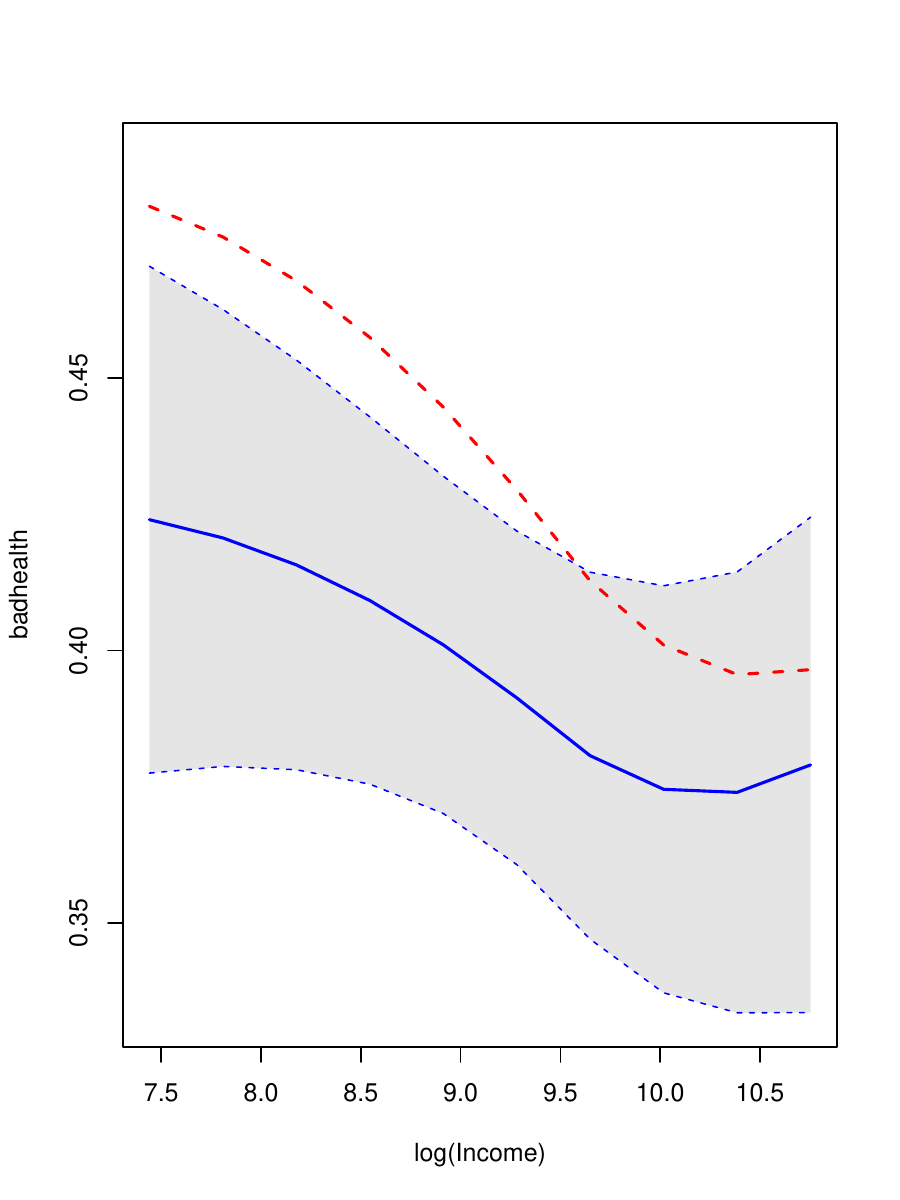}
	\caption{The solid line depicts the FPW series estimator $\widehat g$ while the dashed line depicts a MAR series estimator. The range is from $7.46$ which is the $5\%$ quantile of observed $logIncome$ (i.e., the $\log$ of 1740 Euros) and $10.78$ which is the $95\%$ quantile of observed $logIncome$ (i.e., the $\log$ of 48000 Euros).}\label{health.fig}
\end{figure}

Figure \ref{health.fig} depicts the FPW series estimator with the 95\% uniform confidence bands together with the MAR series estimator evaluated at the median age.
The uniform confidence bands are computed using the bootstrap procedure as described in Section \ref{subsec:ucb} with 1000 bootstrap iterations.
For the MAR series estimator, we consider listwise deletion of missing values.
The FPW series estimator shows a negative association between income and bad health measured by the grip strength which is moderate. For example we find, that the risk of bad health for individuals  at the 25th percentile of income (about 8.9 log income) is about 40\% while for individuals at the 75th percentile (log income of 9.9) it is slightly lower (about 37.5\%). However, according to the confidence interval this difference is not significant. For higher incomes the risk only changes moderately and changes are again not statistically different.  Importantly, our analysis shows that the MAR assumption leads to biased results and potentially erroneous conclusions about the association between income and health. The negative association obtained in the MAR estimator is far more pronounced than in the estimator which accounts for the nonrandom nonresponse. Specifically, with missing at random we predict a risk of bad health at the 25th percentile of about 44\% which is only close to 40\% at the 75th percentile of the income distribution. Note, the confidence intervals show that the results of the two different estimators are significantly different for incomes below the median. For higher incomes the differences are not significantly different.

In addition to the non-linear analysis, we assume a linear $g$ function and present results from a linear model. This linear specification has been used in the literature to test the absolute income hypothesis derived in \cite{Pre_1975}, see e.g. \cite{adeline2017}.
\begin{table}[ht!]
\begin{center}
{\small \begin{tabular}{@{\extracolsep{3pt}}lccc}
\\[-1.8ex]\hline
\hline \\[-1.8ex]
 & \multicolumn{3}{c}{\textsc{Probability of Bad health}} \\
\\[-1.8ex] & \textsc{FPW} & \textsc{MAR} &\textsc{IPW}\\
\hline \\[-1.8ex]
 Constant & -0.104 &  0.004&-0.164$^{**}$\\
  								& (0.078) 							& (0.085) &(0.066)\\
  & & &\\
 log(Income) &-0.009 & -0.017$^{**}$& -0.006\\
  									& (0.006) 						& (0.007)  &(0.005)\\
  & & &\\
 Gender & -0.382$^{***}$ &  -0.413$^{***}$&-0.395$^{***}$\\
  							& (0.013) 						&  (0.014)&(0.013)\\
  & & &\\
 Age 			& 0.009$^{***}$ & 0.009$^{***}$ &0.010$^{***}$\\
  							& (0.001) 					&  (0.001)&(0.001)\\
  & & &\\
  \hline \\[-3ex]
  \end{tabular}}
  \caption{{\small OLS results for FPW, MAR and IPW. Note: $^{*}$p$<$0.1; $^{**}$p$<$0.05; $^{***}$p$<$0.01}}\label{lin:tab}
  \end{center}
  \end{table}

In Table \ref{lin:tab} we depict the results using ordinary least squares estimators with and without probability weighting to account for nonrandom nonresponse. For the FPW estimator we leave the functional form of the selection probability completely unrestricted. Overall, this application underlines the importance to account for nonrandom nonresponse in income information when studying the link between income and health. Importantly, while the MAR finds a negative relation between bad health and income which is significant at the 5\% level, the estimators which account for nonrandom nonresponse reject a significant relation between bad health and income in a linear model. Finally, we note that overall the FPW estimator does not lead to larger standard errors relative to MAR, even if the selection probability is estimated via nonparametric instrumental variable method.

\subsection{Application II: earnings and the demand for housing}\label{s_housing}

In the second application we revisit the question how income affects the demand for housing, for previous studies see e.g.  \cite{QuiRap2004}, \cite{Albouyetal2016} or \cite{DusFitZim2018}. For example, \cite{DusFitZim2018} show for Germany that about 70\% of households in the lowest income quintile are renters whereas in the highest quintile the share is with 30\% markedly lower. As in the application of health and income, studies on housing demand are in general based on survey data with self reported income which  suffer from nonrandom selection.  In the following we estimate the relationship between income and the probability to own a house and quantify the bias when not accounting for nonrandom nonresponse.

The empirical analysis is based on the data of the SOEP. The SOEP is a longitudinal household survey of the German population, for more information see \cite{WagFriSch_07}. The survey includes  self reported standard socio-demographic characteristics including housing and information about different income measures. In contrast to the previous application we focus on a narrow definition of income, labor earnings, which is the most important income component for most individuals. Since a larger fraction of women does not have positive labor earnings, we restrict the analyses to men.

The individual SOEP data can be linked to regional data with information about the average socio-economic situation at the ZIP-code level or even at the residential block. The regional data is provided by a private marketing company which uses administrative information from tax records in combination with credit card information and information about local infrastructure, for more details see \cite{Goebel_etal_2014}. From the regional data, we use the information about the average purchasing power of households living in a specific residential block to construct an instrument for potentially non-random missings of the self-reported earnings information. This information is well suited to construct an instrument. First there exists a strong positive correlation (0.3025) between the individual labor earning and the average purchasing power of households living in a specific residential block. Second, the regional information is available for all individuals such that the instrument does by definition not suffer from nonresponse.

In the empirical analysis we concentrate on $11735$ employed men which are younger than 65 years. Table \ref{table:sum} provides summary statistics of the relevant variables for the analysis.
\bigskip

\begin{table}[h]
\renewcommand{\arraystretch}{1.3}
\centering
{\small
\begin{tabular}{lrrrr}
\hline\hline
 & 1st Qu. & Median & Mean & 3rd Qu.   \\ \hline
\text{Rate of home owner}    & 0  & 0 & 0.45 & 1  \\
\text{Log Gross Earnings per Month}  & 7.49  & 7.88 & 7.76 & 8.25  \\
\text{Age in years}    &32.00 &  41.00 &  40.29  & 49.00  \\

  \hline
\end{tabular}%
}
\caption{\small This table provides summary statistics of the relevant variables for the analysis. The sample includes $11735$ individuals which are aged between 16 and 65 years.  }\label{table:sum}
\end{table}

To quantify the association between earnings and home ownership we use again a semiparametric model with a binary indicator of home ownership and age as depend variable and log monthly earnings as explanatory variable.

\begin{align}\label{est_eq2}
Homeowner_i= g\big(\log(Earning_i^*)\big)+\alpha_0 Age_i + U_i,
\end{align}
where the function $g$ and parameter $\alpha_0$  are unknown. $U_i$ is conditional mean independent of the explanatory variables, $\log(Earning_i^*)$ and $Age_i$.

\begin{figure}[ht]
	\centering
		\includegraphics[width=12cm]{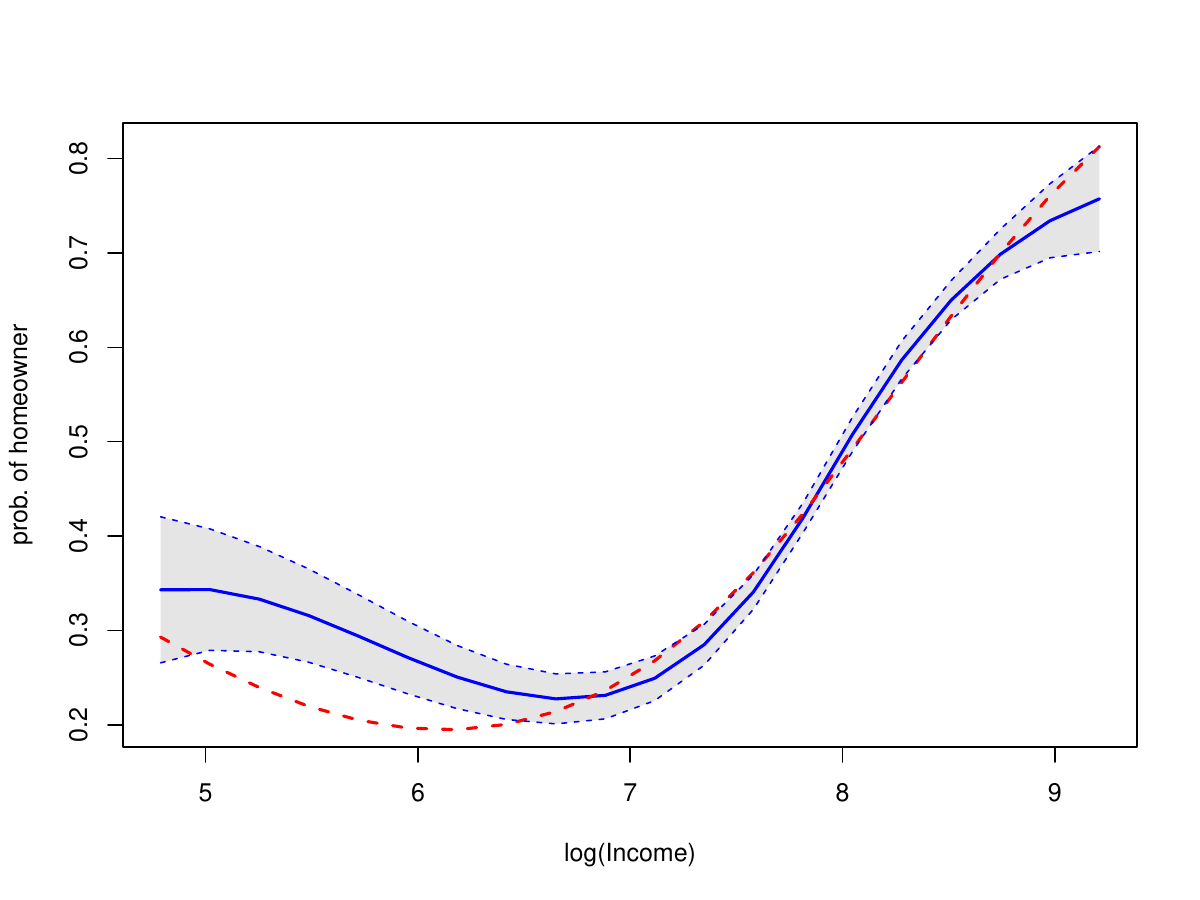}
	\caption{The solid line depicts the FPW series estimator $\widehat g$ while the dashed line depicts a MAR series estimator. The uniform confidence bands are based on 1000 bootstrap iterations. For the MAR series estimator, we consider listwise deletion of missing values.
 }\label{home.fig}
\end{figure}

Figure \ref{home.fig} depicts the FPW series estimator with the 95\% uniform confidence bands (again described in Section \ref{subsec:ucb} using 1000 bootstrap iterations) together with the MAR series estimator evaluated at the median age.
The selection probability $\sol$ is again estimated  using the sieve minimum distance procedure in \eqref{est:prob} with tensor product of cubic B-splines and one knot.
We estimate the function $g$ using the FPW series estimator $\widehat g$ given in \eqref{gen:def:est} using cubic B-splines  with two knots placed at the quantiles of the observed values of
earnings. The MAR estimator is also based on cubic B-splines with the same number and placement of knots.
The FPW series estimator shows a strong positive but non-linear relation between earnings and home ownership. For example we find, that the probability to own a house amounts to about 30\% at monthly earnings below the median earnings in the sample (about 2400 Euros per months). The probability markedly increases to 60\% at monthly earnings at the 75\% percentile (8.25 log points or about 3800 Euro per months) and further to close to 70\% at earnings above 6500 Euros (8.8 log points).  Importantly, our analysis shows that the MAR assumption might lead to biased results  about the association between earnings and the home ownership rate for individuals with very low earnings, i.e. individuals in the lowest decile, which is a key group for public policy. Specifically, we find significantly lower probabilities of home ownerships when not accounting for nonrandom nonresponse. The difference is with about 10 percentage points economically important. Interestingly, for the rest of the earnings distribution the two estimators do not significantly differ and point at a positive relationship between earnings and home ownership

Finally, we assume again a linear $g$ function and consider a linear model (see Table~\ref{lin:housing_tab}). Since the missing at random assumption only affected a small part of the earnings distribution in the nonlinear application, it is not surprising that the coefficients in the linear model do not significantly differ. As expected we find a slightly larger point estimator (0.127) when assuming that nonresponse is random.

\begin{table}[ht!]
\begin{center}
{\small \begin{tabular}{@{\extracolsep{3pt}}lccc}
\\[-1.8ex]\hline
\hline \\[-1.8ex]
 & \multicolumn{3}{c}{\textsc{Probability of home ownership}} \\
\\[-1.8ex] & \textsc{FPW} & \textsc{MAR} &\textsc{IPW}\\
\hline \\[-1.8ex]
 Constant & -0.787$^{***}$ &  -0.853$^{***}$&-0.766$^{***}$\\
  								& (0.041) 							& (0.043) &(0.042)\\
  & & &\\
 log(Earnings) &0.118$^{***}$ & 0.127$^{***}$& 0.117$^{***}$\\
  									& (0.006) 						& (0.006)  &(0.006)\\
  & & &\\
 Age 			& 0.007$^{***}$ & 0.008$^{***}$ &0.008$^{***}$\\
  							& (0.001) 					&  (0.001)&(0.001)\\
  & & &\\
  \hline \\[-3ex]

  \end{tabular}}
  \caption{{\small OLS results for FPW, MAR and IPW. Note: $^{*}$p$<$0.1; $^{**}$p$<$0.05; $^{***}$p$<$0.01}}\label{lin:housing_tab}
  \end{center}
  \end{table}

\section{Conclusion}

In this paper we derive a nonparametric estimators that addresses the problem of nonrandom selection that can be related to nonrandom selection into treatment programs, selective measurement error or through selective nonresponse or missingness of data. Identification of the regression function relies on instrumental variables that are independent of selection conditional on potential covariates.  We obtain identification of our nonparametric regression function without restricting the selection probability to belong to a parametric class of functions via a novel partial completeness assumption and provide primitive conditions for it. We achieve optimal rates of nonparametric rates of convergence of our estimator.  Moreover, the variance of our estimator is not larger than in the case where the variables are fully observed.

We demonstrate the usefulness and relevance of our method in survey data with nonrandom missingness in two different applications with different instruments. First, we analyze the association between bad health and income. We show that standard methods that do not account for the nonrandom selection process are strongly upward biased for individuals with below-median income. Moreover, in a linear model the standard estimator finds a negative relation which is significant at the 5\% level, however the estimators which account for nonrandom non-response reject a significant relation between bad health and income. In the second application we focus on the relation between housing and earnings. While the different estimators with and without the assumption of random missingness lead to the similar estimates in a linear model and for a large share of the earnings distribution, we document significant and important differences for individuals in the lowest earnings decile, which is a central group for public policy.

\appendix
\section{Appendix}\label{app:proofs}
 \begin{proof}[\textsc{Proof of Proposition \ref{prop:prim}.}]
Consider the functional form restriction \eqref{eq:mult}. 
The conditional mean restriction $\Ex[\phi_1(Y)\phi_2(X^*)|Y,W]=1$ and the conditional independence assumption $f_{X^*|Y,W}=f_{X^*|W}$ yield
\begin{align*}
\Ex[\phi_2(X^*)|W=w]-\frac{1}{\phi_1(y)}=0\text{ and }\Ex[\phi_2(X^*)|W=w]-\frac{1}{\phi_1(y')}=0
\end{align*}
for all $y,y'$. Subtracting both equations gives $\phi_1(y)=\phi_1(y')$ for all $y,y'$ in the support of $Y$.

Consider the functional form restriction \eqref{eq:tf}. 
 From $\Ex[\psi(\phi_1(Y)+\phi(X^*))|Y,W]=1$ and the restriction $f_{X^*|Y,W}=f_{X^*|W}$ we infer
\begin{align*}
\Ex[\psi(\phi_1(y)+\phi_2(X^*))|W=w]-\Ex[\psi(\phi_1(y')+\phi_2(X^*))|W=w]=0
\end{align*}
for all $y,y'$ in the support of $Y$. 
Given $y'$, we denote by $\underline j\geq 1$ the smallest integer such that $\Ex[\psi^{(\underline j)}\left(\phi_1(y')+\phi_2(X^*)\right)|W=w]\neq 0$ for some $w$. 
Taylor series expansion applied to the analytic function $\psi$ and Fubini's theorem yield
\begin{align*}
\sum_{j\geq \underline j}\frac{\Ex[\psi^{(j)}\left(\phi_1(y')+\phi_2(X^*)\right)|W=w]}{j!}\left(\phi_1(y)-\phi_1(y')\right)^j
=0
\end{align*}
Consequently, we have $\phi_1(y)=\phi_1(y')$ and, since the choice of $(y,y')$ is arbitrary, the function $\phi_1$ is constant, which completes the proof. 
\end{proof}

\begin{proof}[\textsc{Proof of Theorem \ref{thm:ident}.}]
For the proof of the result, we proceed in two steps.
\noindent\textbf{ Step 1:} We show that the regression function $g$ satisfies equation \eqref{eq:ident:g}.
Making use of relation $f_{Y|X^*}=f_{Y\Delta |X^*}/f_{\Delta | Y X^*}$ and $f_{Y\Delta |X^*}=f_{Y |\Delta X^*}f_{\Delta |X^*}$
we obtain
\begin{align}
g(x)&=\int yf_{ Y |X^*}(y|x)dy\nonumber\\
&=\int y\underbrace{\frac{f_{\Delta|X^*}(1|x)}{f_{\Delta|YX^*}(1|y,x)}}_{=\omega(y,x)} f_{ Y |\Delta X^*}(y|1,x)dy\nonumber\\
&=\Ex\left[Y\omega(Y,x)\Big|\Delta=1, X^*=x\right]\label{eq:ident:proof}
\end{align}
using the definition of the FPW function $\omega$ as given in \eqref{def:omega}.

\noindent\textbf{Step 2:} We show that the FPW function $\omega$ in \eqref{def:omega} is identified. 
Consider the identified set of functions given by
\begin{align*}
\mathcal I = \Big\{ \phi\in\mathcal B:  \Ex\left[\left.D/\phi(Y,X^*)\right|Y,W\right]  = 1\Big\}.
\end{align*}
Clearly, for the true selection probability $\varphi(y,x)=\PP(\Delta=1|Y=y,X^*=x)$ we have $\varphi\in\mathcal I$ by Assumption \ref{A_pos} which ensures that $\varphi$ is uniformly bounded from below.  Assumption \ref{A_instruments} implies for any function $\phi\in\mathcal I$ that
\begin{align*}
  \Ex\left[\left.\frac{\varphi(Y,X^*)}{ \phi(Y,X^*)}-1\right|Y,W\right] & = 0.
\end{align*}
Partial completeness, see Assumption \ref{A_identification}, yields
\begin{align*}
\frac{\varphi(Y,X^*)}{ \phi(Y,X^*)}-1=\psi(X^*)
\end{align*}
for some unknown function $\psi$ and thus,
\begin{align}\label{cond_prob_yx}
 \varphi(Y,X^*)=\phi(Y,X^*)(\psi(X^*)+1).
\end{align}
Employing relation \eqref{eq:ident:proof} with $Y\omega(Y,x)$ replaced by $\omega(Y,x)$, we obtain
\begin{align*}
1=\Ex\left[\omega(Y,x)\Big|\Delta=1, X^*=x\right]
\end{align*}
and hence, by the definition of FPW function $\omega$, the conditional probability $\PP(\Delta=1|X^*=x)$ satisfies
\begin{align}\label{cond_prob_x}
\PP(\Delta=1|X^*=x)=\left(\Ex \left[\frac{1}{\varphi(Y,x)}\Big|\Delta=1, X^*=x\right]\right)^{-1}.
\end{align}
Consequently, we obtain for the FPW function $\omega$ that
\begin{align*}
\omega(y,x)&=\frac{\PP(\Delta=1|X^*=x)}{\varphi(y,x)}\qquad\qquad\qquad\quad\, \text{(due to Definition in equation \eqref{def:omega})}\\
&=\left(\Ex \left[\frac{\varphi(y,x)}{\varphi(Y,x)}\Big|\Delta=1, X^*=x\right]\right)^{-1}
\qquad\qquad\qquad \text{(due to equation \eqref{cond_prob_x})}\\
&=\left(\Ex \left[\frac{\phi(y,x)(\psi(x)+1)}{\phi(Y,x)(\psi(x)+1)}\Big|\Delta=1, X^*=x\right]\right)^{-1}
\quad\,\,\text{(due to equation \eqref{cond_prob_yx})}\\
&=\left(\Ex \left[\frac{\phi(y,x)}{\phi(Y,x)}\Big|\Delta=1, X^*=x\right]\right)^{-1}\\
\end{align*}
for all $\phi\in\mathcal I$. This shows identification of the FPW function $\omega$ which thus completes the proof. 
\end{proof}

\paragraph{Additional Notation} For ease of notation, let $V_i=(Y_i,X_i)$ for $1\leq i\leq n$. Let $\cF$ be a class of measurable functions with a measurable envelope function $F$. Then $N(\varepsilon, \cF , \|\cdot\|_V)$ and $N_{[\,]}(\varepsilon,\cF , \|\cdot\|_V)$, respectively, denote the covering and bracketing numbers for the set $\cF$. The bracketing integral of $\cF$ is denoted by
\begin{align*}
J_{[\,]}(1,\cF,L^2_{V})=\int_0^1\sqrt{1+\log N_{ [\,]}(\varepsilon \,\|F\|_V,\cF , \|\cdot\|_V)}d\varepsilon.
\end{align*}

For ease of notation we write $\sum_i$ for $\sum_{i=1}^n$ and $\sum_{\Delta_i}$ for $\sum_{i=1, \Delta_i=1}^n$. 
We further define $\widehat Q=n^{-1}\sum_{\Delta_i} p^K(X_i)p^K(X_i)'$ and 
$\widehat h(x,\phi)=p^K(x)'\,(n\widehat Q)^{-1}\,\sum_{\Delta_i}p^K(X_i) /\phi(V_i)$. We thus have that the estimator of the FPW function coincides with $\widehat\omega(v,\phi)=\big(\phi(y,x)\widehat h(x,\phi)\big)^{-1}$. 
Further, let $\widehat Q(\phi)=n^{-1}\sum_{\Delta_i}\widehat \omega(V_i,\phi) p^K(X_i)p^K(X_i)'$.
By Assumption \ref{Ass_bas}, the eigenvalues of $\Ex[\Delta\, p^K(X)p^K(X)']$ are bounded away from zero and hence, it may be assumed  that $\Ex[\Delta\, p^K(X)p^K(X)']=I_K$, where $I_K$ denotes the $K\times K$ identity matrix. We also denote $\gamma=\Ex[ \Delta g(X)p^K(X)]$. Throughout the proofs, we use the notation $a_n\lesssim b_n$ to denote $a_n\leq C b_n$ for some constant $C > 0$ and for all $n\geq 1$.  

\begin{proof}[\textsc{Proof of Theorem \ref{thm:est:cond:par}.}]
 The proof is based on the upper bound
 \begin{align*}
  \|\widehat g-g\|_X\leq  \|\widehat g-\gamma'p^K\|_X +\|\gamma'p^K-g\|_X.
 \end{align*}
 Since $\|\gamma'p^K-g\|_X=O(K^{-{\alpha/d_x}})$ by Assumption \ref{Ass_bas} $(iv)$ it is sufficient to consider the first term on the right hand side. We observe
\begin{multline*}
 \|\widehat g-\gamma'p^K\|_X^2\lesssim \big\|\widehat Q(\widehat\varphi)^{-1}\|^2\|I_K-\widehat Q(\widehat\varphi)\|^2\big\|n^{-1}\sum_{\Delta_i} p^K(X_i) \big(Y_i-\gamma'p^K(X_i)\big)\, \widehat \omega(V_i,\widehat\varphi)\big\|^2\\\hfill
 +\big\|n^{-1}\sum_{\Delta_i} p^K(X_i) (Y_i-\gamma'p^K(X_i))\, \widehat \omega(V_i,\widehat\varphi)\big\|^2.
\end{multline*}
From Lemma \ref{Lem:Mat} we deduce $\|\widehat Q(\widehat\varphi)-I_K\|^2=K^2/n$ and thus $\|\widehat Q(\widehat\varphi)^{-1}\|^2=1+o_p(1)$.
Consequently, it is sufficient to consider
\begin{align*}
\big\|n^{-1}&\sum_{\Delta_i} (Y_i-\gamma'p^K(X_i))\, \widehat \omega(V_i,\widehat\varphi)p^K(X_i)\big\|^2\\
&\lesssim\underbrace {\big\|n^{-1}\sum_{\Delta_i}Y_i\,\Big(\widehat\omega(V_i,\widehat\varphi)-\omega(V_i)\Big)\,p^K(X_i)\big\|^2}_{I}\\
&+\underbrace {\big\|n^{-1}\sum_{\Delta_i}\gamma'p^K(X_i)\,\Big(\widehat\omega(V_i,\widehat\varphi)-\omega(V_i)\Big)\,p^K(X_i)\big\|^2}_{II}\\
&+\underbrace {\big\|n^{-1}\sum_{\Delta_i}(Y_i-\gamma'p^K(X_i))\, \omega(V_i)p^K(X_i)\big\|^2}_{III}.
\end{align*}
Consider $I$. We have
\begin{align*}
 I&=\big\|n^{-1}\sum_i Y_i\Delta_i\,\big(\frac{1}{\widehat\varphi(V_i)\,\widehat h(X_i,\widehat\varphi)}-\frac{1}{\varphi(V_i)\,h( X_i,\varphi)}\big)\, p^K(X_i)\big\|^2\\
 &\leq \sup_{(\phi,\psi)\in\mathcal B\times\mathcal H}\big\|n^{-1}\sum_iY_i\Delta_i\,\big(\frac{1}{\phi(V_i)\,\psi(X_i,\phi)}-\frac{1}{\varphi(V_i)\,h( X_i,\varphi)}\big)\, p^K(X_i)\big\|^2.
\end{align*}
For $1\leq j\leq K$ and $1\leq i\leq n$, we introduce the function
\begin{equation}\label{def:h_j}
 h_j(V_i,\phi,\psi):=Y_i\Delta_i\,\Big(\frac{1}{\phi(V_i)\,\psi(X_i,\phi)}-\frac{1}{\varphi(V_i)\,h( X_i,\varphi)}\Big)\, p_j(X_i)
\end{equation}
and the class of functions
$\cF_j=\{h_j(\cdot,\phi,\psi):\,(\phi,\psi)\in\mathcal B\times\mathcal H\}$.
For all $(\phi,\psi)\in\mathcal B\times\mathcal H$ we observe
\begin{equation*}
 \big|h_j(V_i,\phi,\psi)\big|
 \leq \big|Y_i\Delta_i\,p_j(X_i)\big|\sup_{(\phi,\psi)\in\mathcal B\times\mathcal H}\big|\frac{1}{\phi(V_i)\,\psi(X_i,\phi)}-\frac{1}{\varphi(V_i)\,h( X_i,\varphi)}\big|=:F_j( V_i)
\end{equation*}
and hence, $F_j$ is an envelope function of the class $\cF_{j}$. 
By the definition of $\mathcal B$, all functions of $\mathcal B$ and thus also of $\mathcal H$ are uniformly bounded away from zero. Therefore, we obtain the upper bound
\begin{align*}
\|F_j\|_V\lesssim \Ex \big[Y^2\Delta\,p_j^2(X)\big]\leq \Ex \big[\Delta g^2(X)\,p_j^2(X)\big]+\sigma^2\,\Ex \big[ p_j^2(X^*) \big],
\end{align*}
for some finite constant  $\sigma>0$ such that $\Ex[\Delta\, U^2|X^*]\leq \sigma^2$.
Theorem 2.14.5 of \cite{Vaart2000}  gives
\begin{align*}
\sum_{j=1}^K&\,\Ex\Big[\sup_{(\phi,\psi)\in\mathcal B\times\mathcal H}\Big|n^{-1/2}\sum_i h_j(V_i,\phi,\psi)-\Ex h_j(V,\phi,\psi)\Big|^2\Big]\\
&\leq\sum_{j=1}^K\Big(\Ex\sup_{(\phi,\psi)\in\mathcal B\times\mathcal H}\Big|n^{-1/2}\sum_i h_j(V_i,\phi,\psi)-\Ex h_j(V,\phi,\psi)\Big|+\|F_j\|_V\Big)^2.
\end{align*}
We further conclude by applying the last display of Theorem 2.14.2 of \cite{Vaart2000}  for $1\leq j\leq K$
\begin{equation}\label{key:ineq}
\Ex\sup_{(\phi,\psi)\in\mathcal B\times\mathcal H}\Big|n^{-1/2}\sum_i h_j(V_i,\phi,\psi)-\Ex h_j(V,\phi,\psi)\Big|\lesssim J_{[\,]}(1,\cH_{j},\|\cdot\|_V)\,\|F_j\|_V.
\end{equation}
Due to Lemma 4.2 $(i)$ of \cite{Chen07} we have uniformly in $j$ that
\begin{equation*}
 \log N_{[\,]}\Big(\varepsilon,\,\cF_{j},\|\cdot\|_V\Big)
\leq\log N\Big(\frac{\varepsilon}{2 C},\mathcal B,\|\cdot\|_V\Big)
+\log N\Big(\frac{\varepsilon}{2 C},\mathcal H,\|\cdot\|_X\Big).
\end{equation*}
Following  \cite[Remark 3]{CLVK03econometrics},  $\log N(\delta,\mathcal B,\|\cdot\|_\infty) \leq C \delta^{-d/\alpha}$ and hence, we obtain $\int_0^1\sqrt{\log N(\varepsilon,\mathcal B, \|\cdot\|_\infty)}d\varepsilon<\infty$ as long as $\alpha>d/2$. Similarly, we conclude $\int_0^1\sqrt{\log N(\varepsilon,\mathcal H, \|\cdot\|_\infty)}d\varepsilon<\infty$.
Now Assumption \ref{Ass_bas} $(v)$ together with the inequality $\sqrt{a+b}\leq\sqrt a+\sqrt b$ for $a,b\geq 0$  implies $J_{[\,]}(1,\cF_{j},L^2_{ V})<\infty$ uniformly in $j$.
Consequently, we have
\begin{align*}
I&\lesssim n^{-1}\sum_{j=1}^K\Big(\Ex\sup_{(\phi,\psi)\in\mathcal B\times\mathcal H}\Big|n^{-1/2}\sum_i h_j(V_i,\phi,\psi)-\Ex h_j(V,\phi,\psi)\Big|+\|F_j\|_V\Big)^2+\|F_j\|_V^2\\
&\lesssim n^{-1}\sum_{j=1}^K\|F_j\|_V^2\\
&\lesssim n^{-1}\sup_{x\in\mathcal X}\|p^K(x)\|^2 \,\Big(\|g\|_X^2+\sigma^2\Big)\\
&=O_p(K/n).
\end{align*}
Consider $II$. It follows $II=O_p(K/n)$ similarly to the upper bound for the term $I$ by making use of the following inequality
\begin{align*}
\sum_{j=1}^K \Ex\big|\gamma'p^K(X) p_j(X)\big|^2&\leq \sup_{x\in\mathcal X}\|p^K(x)\|^2
 \|\gamma'p^K\|_X^2\\
 &=\sup_{x\in\mathcal X}\|p^K(x)\|^2 \|\gamma\|^2\\
 &=O(K),
\end{align*}
where the last equality is due to Assumption \ref{Ass_bas} $(ii)$. 
Consider $III$. From Corollary \ref{coro:fpw} we deduce
\begin{align*}
\gamma&= \Ex[\Delta g(X)p^K(X)]\\
&= \Ex\big[\Delta g(X^*)p^K(X^*)\Ex[\omega(Y,X^*)|\Delta,X^*]\big]\\
&= \Ex\big[\Delta Yp^K(X^*)\omega(Y,X^*)\big]+\Ex\big[\Delta (g(X^*)-Y)p^K(X^*)\omega(Y,X^*)\big]\\
&= \Ex[\Delta Y \omega(V)\, p^K(X)]+\underbrace{\Ex\big[(g(X^*)-Y)p^K(X^*)\PP(\Delta=1|X^*)\big]}_{=0}\\
&= \Ex[\Delta Y \omega(V)\, p^K(X)]
\end{align*}
and consequently 
\begin{align*}
\Ex[\Delta(Y-\gamma'p^K(X))\, \omega(V)p^K(X)]&=\gamma-\underbrace{\Ex[\Delta p^K(X) \omega(V)p^K(X)']}_{=I_K}\,\gamma\\
&=0.
\end{align*}
Using that the FPW function $\omega$ is uniformly bounded from above, we thus conclude
\begin{align*}
\Ex III&= n^{-1}\,\Ex\big\|\Delta\,(Y-\gamma'p^K(X))\, \omega(V)p^K(X)\big\|^2\\
&\leq  2\,n^{-1}\,\sup_{x\in\mathcal X}\|p^K(x)\|^2\sup_{v}|\omega(v)|^2\Big(\Ex\big|\Delta(Y-g(X))\big|^2+\Ex\big|\Delta(g(X)-\gamma'p^K(X))\big|^2\Big)\\
&\leq  2\,n^{-1}\,\sup_{x\in\mathcal X}\|p^K(x)\|^2\sup_{v}|\omega(v)|^2\Big(\Var(U)+\|g-\gamma'p^K\|_X^2\Big)\\
&\lesssim n^{-1}K\, (1+K^{-2\alpha/d_x}),
\end{align*}
again using  Assumption \ref{Ass_bas} $(ii)$ and the approximation error imposed in  Assumption \ref{Ass_bas} $(iv)$, which completes the proof.
\end{proof}

\begin{proof}[\textsc{ Proof of Theorem \ref{thm:inference:par}.}]
We make use of the lower bound of the sieve variance given by
 \begin{align*}
   \textsl{v}_K(x)&=
  p^K(x)'\Ex\Big[p^K(X)\Var\big(U\Delta \, \omega(V)\big|X\big)p^K(X)'\Big]p^K(x)\\\hfill
 &\gtrsim \|p^K(x)\|^2, 
 \end{align*}
which is due to the condition that $\Var\big(U\Delta \, \omega(V)\big|X\big) $ is uniformly bounded from below, see Assumption \ref{A:inf:par}. 
 The proof is based on the relationship
 \begin{align*}
  \widehat g(x)-\gamma'p^K(x)
  &=\underbrace {p^K(x)' \widehat Q(\widehat\varphi)^{-1}\frac{1}{n}\sum_{\Delta_i=1}p^K(X_i)\,U_i \,  \omega(V_i)}_{I}\\
  &+\underbrace {p^K(x)' \widehat Q(\widehat\varphi)^{-1}\frac{1}{n}\sum_{\Delta_i=1}p^K(X_i)\,U_i \, \big(\widehat \omega(V_i,\widehat\varphi) -  \omega(V_i)\big)}_{II}\\
  &+\underbrace {p^K(x)' \widehat Q(\widehat\varphi)^{-1}\frac{1}{n}\sum_{\Delta_i=1}p^K(X_i)\,\big( g(X_i)-\gamma'p^K(X_i)\big)\,  \omega(V_i)}_{III},
 \end{align*}
where we evaluate each summand on the right hand side separately. 
Consider $I$. Following the proof  of Theorem \ref{thm:est:cond:par} we obtain
\begin{align*}
  \sqrt {n/\textsl{v}_K(x)}\,I
  &=\sum_{\Delta_i}\big(n\, \textsl{v}_K(x)\big)^{-1/2}p^K(x)'p^K(X_i)\,U_i \, \omega(V_i)+o_p(1)\\
  &=\sum_{\Delta_i} s_{in}+o_p(1).
 \end{align*}
 In the following, we show that $s_{in}$, $1\leq i \leq n$, satisfy the Lindeberg conditions for the CLT. First, note that $\Ex[s_{in}]=0$  we observe
 \begin{align*}
 \Ex[s_{in}]&=\Ex[\Delta\, U \omega(V)p^K(X)]\\
 &=\Ex[\Delta\, U \omega(Y,X^*)p^K(X^*)]\\
 &= \Ex[\PP(\Delta=1|Y,X^*)\, U \omega(Y,X^*)p^K(X^*)]\qquad\quad\qquad\text{(since $U=Y-g(X^*)$) }\\
 &= \Ex[U\, \PP(\Delta=1|X^*)p^K(X^*)]\qquad\quad\text{(due to definition of FPW function } \omega)\\
 &=0,\qquad\qquad\qquad\qquad\qquad\qquad\qquad\qquad\qquad\quad\text{(using that $\Ex[U|X^*]=0$).}
 \end{align*}
Thus,   $s_{in}$, $1\leq i \leq n$, are centered variables and  by the definition of $\textsl{v}_K(x)$ we have $n\Ex[s_{in}^2]=1$. 
Moreover,  for all $\delta>0$ we observe
\begin{align*}
 \sum_{D_i} \Ex [s_{in}^2\1{\{|s_{in}|>\delta\}}]&=n\delta^2\Ex\big[\big|s_{in}/\delta\big|^2 \1\{|s_{in}/\delta|>1\}\big]\\
&\leq n\,\delta^2\, \Ex|s_{in}/\delta|^4 \\
&\leq Cn^{-1} \delta^{-2} K^2 \Ex|\Delta\,U \omega(Y,X^*)|^4\\
&=o(1)
\end{align*}
due to the fourth moments condition imposed in  Assumption \ref{A:inf:par} and the rate condition $K^2=o(n)$. The  Lindeberg-Feller
CLT thus implies $\sum_i s_{in}\stackrel{d}{\rightarrow}\mathcal N(0,1)$.
Consider $II$. Recall that $\omega(V,\varphi)=(\varphi(V)h(X,\varphi))^{-1}$ is identified due to Theorem \ref{thm:ident}. Given consistency of the estimator $\widehat \varphi(\cdot)\widehat h(\cdot,\widehat \varphi)$ it is sufficient to consider the shrinking function class
$\mathcal A_n=\set{(\phi,\psi)\in\mathcal B\times\cH: \|\phi(\cdot)\psi(\cdot,\phi)-\varphi(\cdot)h(\cdot,\varphi)\|_\infty\leq r_n}$ with $r_n=o(1)$.
Recall the definition of $h_j$,  $1\leq j\leq K$,  in \eqref{def:h_j} and thus
\begin{align*}
 II
 &\leq \sum_{j=1}^K \sup_{(\phi,\psi)\in\mathcal A_n}\Big|\frac{1}{n}\sum_ip_j(x)h_j(V_i,\phi,\psi)\Big|+o_p\big(\sqrt{\textsl{v}_K(x)}\big).
\end{align*}

For all $(\phi,\psi)\in\mathcal A_n$ we observe
\begin{equation*}
 \big|h_j(V_i,\phi,\psi)\big|
 \leq \big|Y_i\,p_j(X_i)\big|\sup_{(\phi,\psi)\in\mathcal A_n}\big|\frac{1}{\phi(V_i)\,\psi(X_i,\phi)}-\frac{1}{\varphi(V_i)\,h( X_i,\varphi)}\big|=:F_{jn}( V_i)
\end{equation*}
and hence, $F_{jn}$ is an envelope function of the class $\cF_{jn}=\{h_j(\cdot,\phi,\psi):\,(\phi,\psi)\in\mathcal A_n\}$. In particular, using that $\varphi(\cdot)\,h( \cdot,\varphi)$ is uniformly bounded from below (by the definition of the function class $\mathcal B$) we obtain by employing inequality \eqref{key:ineq} that
\begin{align*}
\sqrt n \Ex|II|&\lesssim \sum_{j=1}^K\sqrt{\Ex|p_j(x)F_{jn}(V)|^2}\\
&\lesssim\sup_{(\phi,\psi)\in\mathcal A_n}\|\phi(\cdot)\psi(\cdot,\phi)-\varphi(\cdot)h(\cdot,\varphi)\|_\infty\sum_{j=1}^K\sqrt{\Ex|p_j(x)p_j(X)|^2}\\
&\lesssim r_n\sum_{j=1}^K\sqrt{\Ex|p_j(x)p_j(X)|^2}\\
&= r_n\sum_{j=1}^K|p_j(x)|\\
&=o\big(\sqrt{\textsl{v}_K(x)}\big),
\end{align*}
where the last bound is due to Assumption \ref{A:inf:par} $(ii)$. 
Consider $III$. Using $\Ex[\omega(V)|X^*,D=1]=1$ we obtain $\Ex\big[p^K(X)\,\Delta\big( g(X)-\gamma'p^K(X)\big)\,  \omega(V)\big]=0$. We thus have
\begin{align*}
\sqrt n\,III&=n^{-1/2} \|p^K(x)\| \sqrt{\Ex\big\|p^K(X)\,( g(X)-\gamma'p^K(X))\big\|^2}\times o_p(1)\\
&=O_p\Big( n^{-1/2}K \|g-\gamma'p^K\|_\infty \Big)\\
&=o_p(1).
\end{align*}
Consequently, condition \eqref{cond:inference:par} implies  $\sqrt{n/\, \textsl{v}_K(x)}\,\big(\widehat g(x) - g(x)\big)\stackrel{d}{\rightarrow}\mathcal N(0,1)$. The result follows by Lemma \ref{lem_cons_var} which establishes consistency of the sieve variance $\widehat{\textsl{v}}_K(x)$. 
\end{proof}

 \begin{proof}[{\textsc{Proof of Theorem \ref{thm:bands}.}}]
Due to the \cite[Proof of Theorem 4.1]{chen2018optimal} it is sufficient to show
\begin{align*}
\left|\sqrt{n/\widehat{\textsl{v}}_K(x)}\big(  \widehat g(x) - g(x)\big) - \mathbb{Z}(x)\right|=o_p(r_n)
\end{align*}
since then the result follows by the anti-concentration inequality of \cite[Theorem 2.1]{chernozhukov2014}. 
We denote $Z^n = \{(\Delta_1,Y_1,X_1,W_1), \ldots, (\Delta_n,Y_n,X_n,W_n)\}$.
\\
\textbf{Step 1.} We start by showing that $\sqrt{n/\widehat{\textsl{v}}_K(x)}\big( \widehat g(x) - g(x)\big)$ can be uniformly approximated by the process
  \begin{align*}
 \widehat{ \mathbb Z}(x) =  \frac{p^K(x)'}{\sqrt{n{\textsl{v}}_K(x)}}\sum_{\Delta_i}^np^K(X_i)U_i \,  \omega(V_i).
\end{align*}
We observe
{\small \begin{align*}
  &\left|\sqrt{n/\widehat{\textsl{v}}_K(x)}\big(  \widehat g(x) - g(x)\big) - \widehat{\mathbb{Z}}(x)\right| \\
  &\leq 
  \underbrace{\left|\frac{\sqrt{n}p^K(x)'\widehat Q(\widehat\varphi)^{-1}}{\sqrt{ \textsl{v}_K(x)}}\sum_{\Delta_i}p^K(X_i)\,U_i \,  \omega(V_i) - \widehat{\mathbb{Z}}(x)\right|}_{I(x)}\\
&  + \left|\sqrt{\frac{ \textsl{v}_K(x)}{\widehat{\textsl{v}}_K(x)}} - 1\right| \Bigg(\underbrace{
\left|\frac{\sqrt{n}p^K(x)'Q^{-1}}{\sqrt{ \textsl{v}_K(x)}}\frac{1}{n}\sum_{\Delta_i}p^K(X_i)\,U_i \,  \omega(V_i) \right|}_{II(x)}\\
&\qquad+  \underbrace{
\left|\frac{\sqrt{n}p^K(x)'Q^{-1}}{\sqrt{ \textsl{v}_K(x)}}\frac{1}{n}\sum_{\Delta_i}p^K(X_i)\,U_i \, \big(\widehat \omega(V_i,\widehat\varphi) -  \omega(V_i)\big)\right|
 }_{III(x)}\\
 &\qquad+\underbrace{\left|\frac{\sqrt{n}p^K(x)'Q^{-1}}{\sqrt{ \textsl{v}_K(x)}}\frac{1}{n}\sum_{\Delta_i}p^K(X_i)\,\big( g(X_i)-\gamma'p^K(X_i)\big)\,  \omega(V_i)\right|}_{IV(x)}\\
  &\qquad+\underbrace{\left|\frac{\sqrt{n}}{\sqrt{ \textsl{v}_K(x)}}\Big(p^K(x)'Q^{-1}\frac{1}{n}\sum_{\Delta_i}p^K(X_i)\,\gamma'p^K(X_i)\,  \omega(V_i)-g(x)\Big)\right|}_{V(x)}\Bigg).
\end{align*}}
We have $\| \widehat Q(\widehat\varphi) - \text{I}_K\| = O_p(K\sqrt{\log(n)/n})$, see Lemma \ref{Lem:Mat}. Further, we obtain
\begin{align*}
\sup_{x\in\mathcal{C}}I(x)&=\sup_{x\in\mathcal{C}}\left|\frac{p^K(x)'(\widehat Q(\widehat\varphi)^{-1}-\text{I}_K)}{\sqrt{n\textsl{v}_K(x)}}\sum_{\Delta_i}p^K(X_i)\,U_i \,  \omega(V_i)\right|\\
 &=O_p(K^2\sqrt{\log(n)/n}).
\end{align*}
Define the process $\mathbb Z(x)=p^K(x)'Q^{-1/2}\mathcal Z/\sqrt{\textsl{v}_K(x)}$. We have
\begin{align*}
 \sup_{x\in\mathcal{C}} II(x) & \leq   \sup_{x\in\mathcal{C}}I(x)+ \sup_{x\in\mathcal{C}}|\widehat{\mathbb Z}(x)|\\
  & =  O_p\Big(K^2\sqrt{\log(n)/n}\Big) + \sup_{x\in\mathcal{C}}\big|\widehat{\mathbb Z}(x) - \mathbb Z(x)\big| + \sup_{x\in\mathcal{C}}|\mathbb Z(x)|\\
  & =   O_p\Big(K^2\sqrt{\log(n)/n}\Big) + o_p(r_n) + \sup_{x\in\mathcal{C}}|\mathbb Z(x)|\\
  & =  o_p(r_n) + O_p(c_n).
\end{align*}
where the third bound is due to step 2 below and the last equality is because of the condition $K^{5/2}=o(r_n^3\sqrt{n})$ and by \cite[Lemma G.5]{chen2018optimal}, which is valid under our assumptions and which implies $\sup_{z\in\mathcal{C}}|\mathbb{Z}(z)| = O_p(c_n)$.
Consider  $III(x)$. Recall the definition of $h_j$,  $1\leq j\leq K$ given in \eqref{def:h_j}, we obtain
\begin{align*}
 & \sup_{x\in\mathcal{C}}\frac{\sqrt{n}\left|III(x) \right|}{\sqrt{\textsl{v}_K(x)}}\\
&=\sup_{x\in\mathcal{C}}\frac{\|p^K(x)'\widehat Q(\widehat\varphi)^{-1}\|}{\sqrt{\textsl{v}_K(x)}}\Big(\sup_{(\phi,\psi)\in\mathcal B\times\mathcal H}\sum_{j=1}^K\Big| \frac{1}{\sqrt n}\sum_i\psi_j(V_i,\phi,\psi)\Big|^2\Big)^{1/2}\\
&=o_p(K^2\sqrt{\log(n)/n})
\end{align*}
following the proof of Theorem \ref{thm:inference:par}. Moreover, we observe
\begin{align*}
&\sup_{x\in\mathcal{C}}\frac{\sqrt{n}\left|IV(x) \right|}{\sqrt{ \textsl{v}_K(x)}}\\
 &\leq 
\sup_{x\in\mathcal{C}}\frac{\|p^K(x)'\widehat Q(\widehat\varphi)^{-1}\|}{\sqrt{ \textsl{v}_K(x)}} \Big\| \frac{1}{\sqrt n}\sum_{\Delta_i}p^K(X_i)\,\big( g(X_i)-\gamma'p^K(X_i)\big)\,  \omega(V_i)\Big\|\\
&=O_p\Big( n^{-1/2}K \|g-\gamma'p^K\|_\infty \Big)
\end{align*}
again following the proof of Theorem \ref{thm:inference:par}. 
 For the last summand we note
\begin{align*}
 \sup_{x\in\mathcal{C}}\frac{\sqrt{n}\left|V(x) \right|}{\sqrt{ \textsl{v}_K(x)}} &\leq \sup_{x\in\mathcal{C}}\frac{\sqrt{n}}{\sqrt{\textsl{v}_K(x)}}\big|\gamma'p^K(x)-g(x)\big|.
\end{align*}
Consequently, Lemma \ref{lem_cons_var}, i.e.,  $\sup_{x\in\mathcal{C}}\left|\sqrt{\textsl{v}_K(x)/\widehat{\textsl{v}}_K(x)} - 1\right| = O_p(\sqrt{n^{-1/2}K^{1/2}\log(n)})$ and
the rate requirement in Assumption \ref{Ass:uniform} $(iii)$ imply
\begin{align*}
\left|\sqrt{n/\widehat{\textsl{v}}_K(x)}\big(  \widehat g(x) - g(x)\big) - \widehat{\mathbb{Z}}(x)\right|=o_p(r_n).
\end{align*}

\noindent\textbf{Step 2.} 
We have
\begin{align*}
\sum_i\Ex&\left\|\frac{1}{\sqrt n}p^K(X_i)U_i \,  \omega(V_i)\right\|^3\\
&\lesssim \frac{K^{3/2}}{\sqrt n}.
\end{align*}
 Further,  recall that $r_n$ is a sequence satisfying
 \begin{align*}
 \frac{K^{5/2}}{ r_n^3\sqrt{n}}=o(1).
 \end{align*}
Hence we may apply Yurinskii's coupling (\cite[Theorem 10]{2002Pollard}) and consequently,  there exists a sequence of $\mathcal N(0,\Sigma)$ distributed random vectors $\mathcal Z$ such that
\begin{equation}\label{proof:Th:3:4:step:2_1}
  \left\|\frac{1}{\sqrt n}p^K(X_i)U_i \,  \omega(V_i) - \mathcal Z\right\| = o_p(r_n).
\end{equation}
Recall the definition $\mathbb Z(x)=p^K(x)'Q^{-1/2}\mathcal Z/\sqrt{\textsl{v}_K(x)}$, which is a centered Gaussian process with covariance function 
\begin{align*}
\Ex[\mathbb Z(x_1)\mathbb Z(x_2)] = p^K(x_1)'Q^{-1/2}\, \Sigma\, Q^{-1/2} p^K(x_2)\Big/\sqrt{ \textsl{v}_K(x_1)\textsl{v}_K(x_2)}.
\end{align*} 
Hence, by equation \eqref{proof:Th:3:4:step:2_1} we have
\begin{equation}\label{eq_step_2_UCB_exogenous}
  \sup_{x\in\mathcal{C}}\left|\widehat{\mathbb Z}(x) - \mathbb Z(x)\right| = o_p(r_n).
\end{equation}

\noindent\textbf{Step 3.} In this step we approximate the bootstrap process by a Gaussian process. Under the bootstrap distribution $\mathbb{P}^*$ each term $Y_i\,\widehat \omega(Y_i,X_i;\widehat \sol) - \widehat g(X_i)$ has mean zero for all $1\leq i\leq n$. Moreover, we have
\begin{equation*}
\frac{1}{n}\sum_{D_i}\Ex\left[\left.\widehat Q^{-1} p^K(X_i)\Big(Y_i\,\widehat \omega(V_i, \widehat \sol) - \widehat g(X_i)\Big)^2\varepsilon_i^2p^K(X_i)'\widehat Q^{-1}\right|Z^n\right]
= \widehat \Sigma.
\end{equation*}
Since $\Ex[|\varepsilon_i|^3|Z^n] < \infty$ uniformly in $i$, we have
\begin{align*}
\sum_i\Ex&\left[\left.\left\|\frac{1}{\sqrt{n}}p^K(X_i)\Delta_i\Big(Y_i\,\widehat \omega(V_i,\widehat \sol) - \widehat g(X_i)\Big)\varepsilon_i\right\|^3\right|Z^n\right]\\
 &\lesssim \frac{1}{\sqrt{n}}\,\Ex\|p^K(X)\|^2\sup_x\|p^K(x)\| \\
&\lesssim \frac{K^{3/2}}{\sqrt{n}},
\end{align*}
with probability approaching one (wpa1).
Again using \cite[Theorem 10]{2002Pollard}, conditional on the data $Z^n$, implies existence of  a  $\mathcal N(0,\widehat \Sigma)$ distributed random vectors $\mathcal Z^*$ such that
\begin{equation*}
  \left\|\frac{1}{\sqrt{n}}\sum_{D_i}p^K(X_i)\Big(Y_i\,\widehat \omega(V_i,\widehat \sol) - \widehat g(X_i)\Big) - \mathcal Z^*\right\| = o_{p^*}(r_n)
\end{equation*}
wpa1. Therefore, 
\begin{align*}
  \sup_{x\in\mathcal{C}}\left|\mathbb{Z}^B(x) - \frac{p^K(x)'\mathcal Z^*}{\sqrt{\widehat{\textsl{v}}_K(x)}}\right| = o_{p^*}(r_n)
\end{align*}
wpa1. Define a centered Gaussian process $\widetilde{\mathbb{Z}}(\cdot)$ under $\mathbb{P}^*$ as
\begin{align*}
  \widetilde{\mathbb Z}(x) = p^K(x)'Q^{-1/2}\Sigma^{1/2}\widehat \Sigma^{-1/2}\mathcal Z^*/\sqrt{ \textsl{v}_K(x)}
\end{align*}
which has the same covariance function as $\mathbb{Z}(x)$. By \cite[Lemma G.6]{chen2018optimal} below we have:
\begin{align*}
  \sup_{x\in\mathcal{C}}\left|\frac{p^K(x)'Q^{-1/2}}{\sqrt{\widehat{\textsl{v}}_K(x)}}\mathcal Z^* - \widetilde{\mathbb Z}(x)\right| = o_{p^*}(r_n)
\end{align*}
wpa1. This and the previous rate of convergence imply that
\begin{equation*}
  \sup_{x\in\mathcal{C}}\left|\mathbb Z^B(x) - \widetilde{\mathbb Z}(x)\right| = o_{p^*}(r_n)
\end{equation*}
wpa1, which completes the proof. 
\end{proof}

\section{Technical Assertions}\label{app:tech}
\begin{lem}\label{Lem:Mat}
Under the conditions of Theorem \ref{thm:est:cond:par} it holds
\begin{align*}
\|\widehat Q(\widehat\varphi)-Q(\varphi)\|=O_p(K/\sqrt n).
\end{align*}
\end{lem}
\begin{proof}
From \cite{belloni2015} we deduce
\begin{align*}
\|\widehat Q(\widehat\varphi)-Q(\varphi)\|&\leq \|\widehat Q(\widehat\varphi)-\widehat Q(\varphi)\|+\|\widehat Q(\varphi)-Q(\varphi)\|\\
&=\|\widehat Q(\widehat\varphi)-\widehat Q(\varphi)\|+O_p\big(\sqrt{(\log n)K/n}\big).
\end{align*}
Further, as the spectral norm is bounded by the Frobenius norm we have
\begin{align*}
\|\widehat Q(\widehat\varphi)-\widehat Q(\varphi)\|^2\leq \sum_{j,l=1}^K\Big|\frac{1}{n}\sum_{\Delta_i} \big(\frac{1}{\widehat\varphi(V_i)\,\widehat h(X_i,\widehat\varphi)}-\frac{1}{\varphi(V_i)\,h( X_i,\varphi)}\big)\, p_j(X_i)p_l(X_i)\Big|^2,
\end{align*}
where the term on the right hand side is of the order $O_p(K^2/n)$ which is due to the analysis preceding inequality \eqref{key:ineq}. 
\end{proof}

\begin{lem}\label{lem_cons_var}
 Let Assumptions \ref{A_instruments}--\ref{A:inf:par} be satisfied. Then,
\begin{align}
\Big|\sqrt{\frac{\widehat{\textsl{v}}_K(x)}{\textsl{v}_K(x)}} -1 \Big|&=o_p(1),\label{pw:cov:est}\\
\sup_{x}\Big|\sqrt{\frac{\widehat{\textsl{v}}_K(x)}{\textsl{v}_K(x)}} -1 \Big|&=O_p\Big(  \sqrt{n^{-1/2}K
\log(n)}\Big).\label{un:cov:est}
\end{align}
\end{lem}
\begin{proof}
Proof of \eqref{pw:cov:est}. 
Note that it is sufficient to establish $\widehat{\textsl{v}}_K(x)-\textsl{v}_K(x)=o_p(\|p^K(x)\|^2)$. 
We make use of the decomposition 
\begin{align}\label{dec:cov:est}
\widehat{\textsl{v}}_K(x) &-\textsl{v}_K(x)= 
p^K(x)'\,\big(\widehat\Sigma-\widetilde\Sigma\big)p^K(x)
+p^K(x)'\,\big(\widetilde\Sigma-\Sigma\big)p^K(x)+o_p(1)
\end{align}
where 
\begin{align*}
\widetilde \Sigma=\frac{1}{n}\sum_{D_i} p^K(X_j)U_i^2\omega^2(V_i) p^K(X_j)', 
\quad\quad\widehat \Sigma=\frac{1}{n}\sum_{D_i} p^K(X_j)\widehat U_i^2\widehat \omega^2(V_i,\widehat \varphi) p^K(X_j)',
\end{align*}
 and $\Sigma=\Ex\widetilde \Sigma$. 
 We further calculate
\begin{align*}
&|p^K(x)'(\widehat\Sigma-\widetilde \Sigma)p^K(x)|\\
&\leq \Big|p^K(x)'\frac{1}{n}\sum_{D_i}p^K(X_i)\Big(\widehat U_i^2 \widehat\omega^2(V_i;\widehat \varphi)-U_i^2\omega^2(V_i;\varphi)\Big)p^K(X_i)' p^K(x)\Big|\\
&\leq \underbrace{n^{-1}\sum_{D_i} \Big|p^K(x)'\big(\widehat U_i \widehat\omega(V_i;\widehat \varphi)-U_i\omega(V_i;\varphi)\big) p^K(X_i) \Big|^2}_{I}\\
&+ 2\underbrace{\Big|p^K(x)' \frac{1}{n}\sum_{D_i}p^K(X_i)\Big(\widehat U_i \widehat\omega(V_i;\widehat \varphi)- U_i \omega(V_i;\varphi)\Big)U_i \omega(V_i;\varphi)p^K(X_i)' p^K(x)\Big|}_{II}.
\end{align*}
Consider $I$. Using $\widehat U_i \widehat\omega(V_i,\widehat \varphi)-U_i\omega(V_i)=(\widehat g(X_i)-g(X_i))\omega(V_i)+(\widehat\omega(V_i,\widehat \varphi)-\omega(V_i))\widehat U_i$ and that $\omega$ is uniformly bounded from above, we obtain
by following the  proof of Theorem \ref{thm:inference:par} that
\begin{align*}
I&\lesssim
\Ex\big\|p^K(x)' p^K(X)\big\|^2\,\|\widehat g- g\|_\infty^2\\
&\quad+\|p^K(x)\|^2\Big(\sum_{j=1}^K\sup_{(\phi,\psi)\in\mathcal A_n}\Big|n^{-1}\sum_i h_j(V_i,\phi,\psi)-\Ex h_j(V,\phi,\psi)\Big|^2\\
&\qquad\qquad+\sum_{j=1}^K\sup_{(\phi,\psi)\in\mathcal A_n}\big|\Ex h_j(V,\phi,\psi)\big|^2\Big)\\
&=O_p\Big(\|p^K(x)\|^2\big(n^{-1}K
+ \|\gamma'p^K-g\|_\infty^2\big)\Big)\\
&=o_p(\|p^K(x)\|^2),
\end{align*}
using that $\Ex\big\|p^K(x)' p^K(X)\big\|^2\leq K$ and $\|\gamma'p^K-g\|_\infty^2=o(K/n)$. 
Again following the  proof of Theorem \ref{thm:inference:par} and making use of the Cauchy-Schwarz inequality yields
\begin{align*}
I &\leq\sqrt{II}\times\sqrt{n^{-1}\sum_i \Big|p^K(x)'U_i \omega(V_i) p^K(X_i)\Big|^2}\\
&\leq \sqrt{II}\times O_p\big(\|p^K(x)\|\sqrt{\Ex[U^2\omega^2(V)]}\big)\\
&=O_p\Big(\|p^K(x)\|^2 (n^{-1/2}K^{1/2}
+  \|\gamma'p^K-g\|_\infty)\Big)\\
&=o_p(\|p^K(x)\|^2),
\end{align*}
using the upper bound of $II$. 

Finally, we obtain
\begin{align*}
p^K(x)'(\widetilde\Sigma-\Sigma)p^K(x)=O_p\big( \|p^K(x)\|^2K/\sqrt{n }\big)
\end{align*}
which is due to the following calculation
\begin{align*}
&\Ex\|\widetilde\Sigma-\Sigma\|^2\\
&= \Ex\Big\|n^{-1}\sum_ip^K(X_i)\Delta_i U_i^2 \omega^2(V_i;\varphi)p^K(X_i)'- \Ex\big[p^K(X)\Delta U^2 \omega^2(V) p^K(X)'\big]\Big\|^2\\
& \leq n^{-1}\Ex\Big\|p^K(X)\Delta U \omega(V;\varphi)\Big\|^4\\
& \leq n^{-1}\sup_x\|p^K(x)\|^4\Ex|\Delta U \omega(V;\varphi)|^4\\
& \lesssim n^{-1}K^2,
\end{align*}
based on the fourth moment condition imposed in Assumption \ref{A:inf:par}.
This establishes consistency of the sieve variance estimator $\widehat{\textsl{v}}_K(x)$ and hence completes the proof. 

The result \eqref{un:cov:est} follows analogously.
\end{proof}
 \section{Extension to Selectively Missing Outcomes}\label{sec:mis:dep}
In many empirical situations, one may also want to control for selective missingness of the dependent variable. Below,  $\Delta^Y$ and $\Delta^X$ denote missingness indicators for the variables $Y^*$ and $X^*$, respectively.
To account for additional selective missingness of $Y^*$, we generalize Assumption \ref{A_instruments} to the exclusion restriction
\begin{align}\label{cond:ex}
\PP(\Delta^Y=1,\Delta^X=1|Y^*,X^*, W) = \PP(\Delta^Y=1,\Delta^X=1|Y^*,X^*).
\end{align}
The previous exclusion restriction has the interpretation that $W$ has no information on the selection indicators $\Delta^Y$ and $\Delta^X$ that is not captured in $(Y^*,X^*)$. In this section, we make use of the notation $\Delta=\Delta^Y\Delta^X$.
Table \ref{table:exclusion} depicts the exclusion conditions required in different sample selection scenarios.
\begin{table}[ht]	
\begin{center}
\renewcommand{\arraystretch}{1.5}
{\small
\begin{tabular}{|l|c|c|}
\hline
&{ $Y^*$ obs.} & $Y^*$ mis.\\
\hline
{ $X^*$ obs.}&  & $\PP(\Delta^Y=1|Y^*,X, W) = \PP(\Delta^Y=1|Y^*,X)$ \\\hline
$X^*$ mis.& $\PP(\Delta^X=1|Y,X^*, W) = \PP(\Delta^X=1|Y,X^*)$ & $\PP(\Delta=1|Y^*,X^*, W) = \PP(\Delta=1|Y^*,X^*)$\\
\hline
\end{tabular}
}
\end{center}
\caption{Exclusion restrictions for instrument $W$ depending on missingness.}\label{table:exclusion}
\end{table}

Under the exclusion restriction imposed in equation  \eqref{cond:ex} we can write the regression function $g$ as
\begin{align*}
g(x)&=\Ex[Y^*|X^*=x]\\
&=\int yf_{ Y^* |X^*}(y|x)\frac{f_{\Delta|Y^* X^*}(1|y,x)}{f_{\Delta | Y^* X^*}(1|y,x)}\,dy\\
&=\int y\frac{f_{\Delta Y^*|X^*}(1,y|x)}{f_{\Delta | Y^* X^*}(1|y,x)}\,dy\\
&=f_{\Delta^X|X^*}(1|x)\int y\frac{f_{\Delta^Y Y^*|\Delta^X X^*}(1,y|1,x)}{f_{\Delta | Y^* X^*}(1|y,x)}\,dy\\
&=\Ex\left[\frac{\Delta^Y Y^*\PP(\Delta^X=1|X^*)}{\PP(\Delta=1|Y^*,X^*)}\Big|\Delta^X=1, X^*=x\right],
\end{align*}
where the right hand side is identified as long as the fractional probability weight $\PP(\Delta^X=1|X^*=x)/\PP(\Delta=1|Y^*,X^*=x)$ is identified.
In the case of potentially missing dependent variable $Y^*$ and fully observed $X^*$, i.e., $X^*=X$,  we obtain, in particular,
\begin{align*}
g(x)&=\Ex\left[\frac{\Delta^Y Y^*}{\PP(\Delta^Y=1|Y^*,X)}\Big|X=x\right]
\end{align*}
as obtained by \cite{breunig2015} (see also Remark \ref{rem:fpw:ipw} when $X$ is exogenous to selection $D^Y$). 
Table \ref{table:ident} provides explicit forms of the regression function $g$ in different sample selection scenarios.
 \begin{table}[ht]	
\begin{center}
\renewcommand{\arraystretch}{2.5}
{\small \begin{tabular}{|l|c|c|}
\hline
&{ $Y^*$ observed $(Y^*=Y)$} & $Y^*$ selectively missing\\
\hline
{ $X^*$ obs. }&  & $\Ex\left[\frac{\Delta^Y Y^*}{\PP(\Delta^Y=1|Y^*,X)}\Big|X=x\right]$ \\\hline
$X^*$ mis.& $\Ex\left[\frac{Y\PP(\Delta^X=1|X^*)}{\PP(\Delta^X=1|Y,X^*)}\Big|\Delta^X=1, X^*=x\right]$ & $\Ex\left[\frac{\Delta^Y Y^*\PP(\Delta^X=1|X^*)}{\PP(\Delta=1|Y^*,X^*)}\Big|\Delta^X=1, X^*=x\right]$\\
\hline
\end{tabular}}
\end{center}
\caption{Identification results for regression function $g$ under different sample selection scenarios.}
\label{table:ident}
\end{table}
  \bibliography{BiB}
\end{document}